\documentclass[journal,twoside,web]{ieeecolor}
\usepackage{generic}
\usepackage{cite}
\usepackage{textcomp}

\usepackage{bm}

\usepackage{amssymb,latexsym,amsfonts,amsmath,bbm}
\usepackage[cal=cm,bb=pazo]{mathalfa}
\usepackage{mathrsfs}
\usepackage{graphicx}
\usepackage{paralist}
\usepackage{dsfont}
\usepackage{eso-pic}
\usepackage{epsfig}
\usepackage{mathtools}
\usepackage{epstopdf}
\usepackage{lipsum}
\usepackage{cite}
\usepackage{subcaption}
\usepackage{booktabs}
\usepackage{multirow}
\newcommand*{\myalign}[2]{\multicolumn{1}{#1}{#2}}

\usepackage{makecell} 

\usepackage{array}
\newcolumntype{C}[1]{>{\centering\arraybackslash}m{#1}}

\newcommand{\PP}{\mathds{P}}

\usepackage[nocomma]{optidef}

\usepackage{tikz}
\usetikzlibrary{calc,shapes,arrows}

\usepackage{multirow}

\newcommand\thetaTildeI{\left[\begin{array}{c} \mathds I_{n} \\\Phi^\top \end{array}\right]}

\usepackage{algorithm}
\usepackage{algorithmic}

\newcommand{\redsquare}{\tikz\fill[red!70!white] (0,0) rectangle (2mm,2mm);}
\newcommand{\bluesquare}{\tikz\fill[blue!30!white] (0,0) rectangle (2mm,2mm);}
\newcommand{\greensquare}{\tikz\fill[green!30!white] (0,0) rectangle (2mm,2mm);}
\newcommand{\pinksquare}{\tikz\fill[pink!80!white] (0,0) rectangle (2mm,2mm);}

\usepackage{multirow}

\usepackage{tikz}
\usetikzlibrary{calc,shapes,arrows,arrows.meta, positioning}

\newcounter{theorem}
\newcounter{definition}
\newcounter{lemma}
\newcounter{claim}
\newcounter{problem}
\newcounter{proposition}
\newcounter{corollary}
\newcounter{construction}
\newcounter{example}
\newcounter{xca}
\newcounter{comments}
\newcounter{remark}
\newcounter{assumption}

\usepackage{stackengine}
\usepackage{scalerel}

\newcommand{\MATLAB}{\textsc{Matlab}\xspace}

\newcommand\Xright{\overrightarrow{\mathbb{X}}}

\newtheorem{theorem}[theorem]{Theorem}
\newtheorem{lemma}[lemma]{Lemma}

\newtheorem{problem}[problem]{Problem}
\newtheorem{proposition}[proposition]{Proposition}

\newtheorem{definition}[definition]{Definition}

\newtheorem{remark}[remark]{Remark}
\newtheorem{assumption}[assumption]{Assumption}

\numberwithin{equation}{section}

\DeclareFontFamily{U}{stix2bb}{}
\DeclareFontShape{U}{stix2bb}{m}{n} {<-> stix2-mathbb}{}

\usepackage[many]{tcolorbox}
\usetikzlibrary{calc}
\tcbuselibrary{skins}

\newtcolorbox{resp}[1][]{%
	enhanced jigsaw,%
	colback=gray!2!white,%
	colframe=gray!80!black,%
	size=small,%
	boxrule=1pt,%
	halign title=flush center,%
	coltitle=black,%
	breakable,%
	drop shadow=black!50!white,%
	attach boxed title to top left={xshift=1cm,yshift=-\tcboxedtitleheight/2,yshifttext=-\tcboxedtitleheight/2},%
	minipage boxed title=3cm,%
	boxed title style={%
		colback=white,%
		size=fbox,%
		boxrule=1pt,%
		boxsep=2pt,%
		underlay={%
			\coordinate (dotA) at ($(interior.west) + (-0.5pt,0)$);
			\coordinate (dotB) at ($(interior.east) + (0.5pt,0)$);
			\begin{scope}[gray!80!black]
				\fill (dotA) circle (2pt);
				\fill (dotB) circle (2pt);
			\end{scope}
		}%
	},%
	#1%
}

\usepackage{xspace}

\newcommand{\R}{{\mathbb{R}}}

\newcommand{\N}{{\mathbb{N}}}

\def\BibTeX{{\rm B\kern-.05em{\sc i\kern-.025em b}\kern-.08em
		T\kern-.1667em\lower.7ex\hbox{E}\kern-.125emX}}

\usepackage{etoolbox}

\makeatletter
\AtBeginDocument{%
	\pagestyle{headings}
	
	\patchcmd{\@oddhead}{\\[-19pt]}{\\[-8pt]}{}{}%
	\patchcmd{\@evenhead}{\\[-19pt]}{\\[-8pt]}{}{}%
}
\makeatother

\definecolor{blue(ryb)}{rgb}{0.01, 0.28, 1.0}
\definecolor{fashionfuchsia}{rgb}{0.96, 0.0, 0.63}
\makeatletter
\let\NAT@parse\undefined
\makeatother
\usepackage[colorlinks=true, citecolor=fashionfuchsia, linkcolor=blue(ryb), final]{hyperref}
\renewcommand{\theequation}{\arabic{equation}}
\counterwithout*{equation}{section}

\makeatletter
\def\@opargbegintheorem#1#2#3{\textit{#1\ #2} \textit{(#3):}}
\makeatother

\begin{document}
	
\title{Data-Driven Stochastic Control: Foundations and Guarantees}
 \author{Abolfazl Lavaei, \IEEEmembership{Senior Member,~IEEE}
 	\thanks{The author is with the School of Computing, Newcastle University, NE4 5TG Newcastle Upon Tyne, United Kingdom (e-mail: {\tt\small{abolfazl.lavaei@newcastle.ac.uk}.}}
}

\maketitle
\begin{abstract}
This work establishes a step forward in advancing data-driven trajectory-based methods for \emph{stochastic} systems with {unknown} mathematical dynamics. In contrast to scenario-based approaches that rely on independent and identically distributed (i.i.d.) trajectories, this work develops a data-driven framework where each trajectory is gathered over a finite horizon and exhibits temporal dependence, referred to as a \emph{non-i.i.d.} trajectory. To ensure safety of dynamical systems using such trajectories, the current body of literature primarily considers dynamics subject to {unknown-but-bounded} disturbances, which facilitates robust analysis. While promising, such bounds may be violated in practice and the resulting \emph{worst-case} robust analysis tends to be overly conservative. To overcome these key challenges, this paper considers stochastic systems with unknown mathematical dynamics, influenced by process noise with \emph{arbitrary} distributions. In the proposed framework, data is collected from stochastic systems under \emph{multiple} realizations within a finite-horizon experiment, where each realization generates a non-i.i.d. trajectory. Leveraging the concept of stochastic control barrier certificates constructed from data, this work quantifies probabilistic safety guarantees with a certified confidence level. To achieve this, the proposed conditions are formulated as a sum-of-squares (SOS) optimization problem, relying solely on empirical average of the collected trajectories and statistical features of the process noise. The efficacy of the approach has been validated on three stochastic benchmarks with unknown models and arbitrary noise distributions. In one case study, it is shown that while no safety controller exists for the robust analysis of the system under bounded disturbances, the proposed stochastic framework yields a safety controller together with quantified probabilistic safety guarantees.
\end{abstract}


\section{Introduction}\label{Sec: Introduction}

\IEEEPARstart{M}{odern} engineered systems are growing in complexity, with distributed physical systems increasingly integrated with computational elements, often operating under uncertainty. Examples span domains such as autonomous robotics, intelligent transportation networks, and healthcare infrastructure, where such systems are naturally modeled as stochastic and play a vital role in many sectors, particularly those where safety is paramount. Malfunctions in these contexts can lead to serious outcomes, from endangering human lives to incurring substantial economic costs \cite{mcgregor2017analysis}.

Modeling stochastic control systems with high fidelity requires capturing the details of each component. However, such systems often involve noisy processes and tightly integrated nonlinear components, leading to rapidly increasing complexity and making accurate modeling highly challenging. Even when such models can be constructed, their complex structure often hinder their practical use in verification and synthesis, limiting the effectiveness of traditional model-based approaches~\cite{hou2013model}. A notable case in point is the self-driving car industry, where most companies rely on \emph{simulations or physical testing} rather than detailed mathematical representations of vehicle dynamics. 

While precise models are rarely available, the widespread use of affordable sensors has enabled extensive data collection from system behavior, such as human driving data for vehicles. This wealth of behavioral data opens the door to systematic analysis and the design of controllers, \emph{i.e.,} the decision-making software that governs autonomous operation. These developments underscore the increasing need for \emph{data-driven} methods capable of handling systems with unknown dynamics, where insights should be derived from observed input-output interactions.

{\bf Data-Driven Literature.} To address the challenge posed by unknown dynamical systems, the literature has introduced two primary categories of data-driven methods: \emph{indirect} and \emph{direct} approaches~\cite{dorfler2022bridging,martin2023guarantees,markovsky2021behavioral}. Indirect methods, including model-based reinforcement learning~\cite{moerland2023model} and Gaussian process regression~\cite{seeger2004gaussian,lederer2019uniform}, aim to perform \emph{system identification}, constructing approximate models of unknown systems~\cite{hou2013model}. These techniques effectively reconstruct the system’s dynamics from data, enabling the application of classical model-based control once a sufficiently accurate model is obtained. However, their practical utility often hinges on the accuracy and tractability of the identification step. In fact, most system identification tools are designed for linear systems or specific nonlinear classes, limiting their applicability to more complex or highly nonlinear systems \cite{hou2013model}. In such settings, the identification process can be both computationally intensive and time-consuming. Additionally, indirect approaches are inherently two-stage~\cite{dorfler2022bridging}, requiring an initial modeling phase followed by analysis or controller synthesis, further adding to their complexity.

In contrast, direct data-driven approaches, such as model-free reinforcement learning~\cite{ramirez2022model}, iterative feedback tuning~\cite{hjalmarsson1998iterative}, virtual reference feedback tuning~\cite{campi2002virtual},  scenario approach~\cite{calafiore2006scenario,campi2009scenario}, and the trajectory-based approach~\cite{de2019formulas,van2020noisy,ou2025stochastic}, adopt a fundamentally different strategy. Rather than first identifying a model, these methods operate {directly} on system measurements, leveraging data to design controllers or perform analysis without explicitly constructing a mathematical representation of the system. This allows control policies to be learned purely from data, which is especially advantageous when model identification is infeasible due to complexity~\cite{tanaskovic2017data}.  Nonetheless, they also introduce key challenges, including the inherent difficulty in providing formal {performance guarantees} in the absence of a model-based framework.

In the last two decades, the field of direct data-driven optimization has witnessed significant advances, leading to the emergence of promising methods for tackling complex control systems. Among these, the \emph{scenario approach} introduced in~\cite{calafiore2006scenario} stands out as a powerful framework for robust control analysis. This method operates by first solving the problem using data and then mapping the outcome back to the original robust system using intermediate formulations that incorporate chance constraints~\cite{esfahani2014performance,margellos2014road,samari2024data,zaker2025data}. While promising, it imposes a critical assumption: the data must be independent and identically distributed (i.i.d.). This means each sample should come from a separate, independent input-out trajectory~\cite{calafiore2006scenario}, requiring the collection of multiple trajectories, potentially millions in practical scenarios. Consequently, this approach is compatible with {simulator-based} settings, where generating such independent runs is practically achievable.

An alternative and complementary method to the scenario approach is the \emph{non-i.i.d.} trajectory-based approach which relies on a single input-output trajectory collected over a fixed time horizon from the unknown system \cite{de2019formulas,berberich2020data,van2020noisy,ou2025stochastic,monshizadeh2024meta,zaker2025data,ashoori2025physics}. It builds on the notion of \emph{persistent excitation}, where the collected data should satisfy a rank condition for certain system classes, as formalized in (generalization of) Willems’ fundamental lemma~\cite{willems2005note}. If this condition is met, the trajectory is considered persistently excited, meaning it captures rich information about the system’s dynamics to enable rigorous analysis. This makes the trajectory-based approach particularly appealing in practical settings where generating many independent trajectories, as required by the scenario approach, is costly or infeasible. To provide a fair perspective, it should be acknowledged that while the scenario approach demands i.i.d. data, it is applicable to ensure safety for the \emph{general class} of nonlinear systems which are Lipschitz continuous. On the other hand, trajectory-based methods that depend on persistence of excitation are mainly limited to specific system classes, such as nonlinear systems with polynomial structure.

{\bf Core Contributions.} While trajectory-based methods have demonstrated significant potential, the existing literature predominantly addresses systems with unknown dynamics subject to {unknown-but-bounded} disturbances. Inspired by the matrix S-lemma proposed in~\cite{de2019formulas,van2020noisy} for deterministic systems, this work develops a foundational trajectory-based framework for \emph{stochastic} control systems with nonlinear polynomial dynamics, influenced by process noise with arbitrary distributions, marking a key departure from conventional assumptions. 

Within this framework, data is collected from multiple realizations of a stochastic system over a finite experimental horizon. This collected data is then used to establish {probabilistic} safety guarantees through the lens of stochastic control barrier certificates (S-CBC), without requiring any explicit system model. To accomplish this, the proposed conditions are cast as a sum-of-squares (SOS) optimization problem, relying solely on empirical average of trajectory data and statistical properties of the process noise. The proposed method has been validated on  physical stochastic systems with unknown dynamics, demonstrating its effectiveness in offering probabilistic safety guarantees directly from noisy data.

As a related work,~\cite{martin2023gaussian} has also considered nonlinear data-driven control under stochastic noise using both Frequentist and Bayesian formulations together with SOS-based synthesis. Unlike this \emph{indirect} framework, which first constructs a Taylor-based model representation and mainly focuses on robust stabilization, the present work develops a direct trajectory-based approach for stochastic safety synthesis via S-CBC. Moreover, the proposed framework only relies on upper bounds on the first two moments of the process noise and is therefore applicable beyond Gaussian noise settings considered in~\cite{martin2023gaussian}. Nevertheless, the Bayesian treatment in~\cite{martin2023gaussian} can provide a more structured characterization of uncertainty when Gaussian prior/posterior assumptions are appropriate.

\textbf{Literature on (Control) Barrier Certificates.} To ensure the safety of dynamical systems, the literature introduced the concept of control barrier certificates (CBC), which offer formal safety guarantees for complex systems with continuous state spaces (see \emph{e.g.,}~\cite{prajna2004safety, prajna2007framework, wieland2007constructive, ames2019control, xiao2023safe}). Conceptually, a CBC is a Lyapunov-like function: it defines inequality constraints over the state space whose values and temporal evolution are governed by the system’s dynamics. If an appropriate level set of a CBC can separate unsafe regions from all admissible trajectories initiated from a given initial set, the existence of such a function certifies the system’s (possibly probabilistic) safety (see \emph{e.g.,}~\cite{ames2016control, clark2019control,santoyo2021barrier,NEURIPS2020_barrier, lavaei2024scalable,nejati2024context,zaker2024compositional}).

However, most existing CBC-based approaches depend critically on the availability of an accurate mathematical model of the system dynamics. While recent research has made notable progress in developing trajectory-based methods for safety analysis without requiring an explicit model, these methods either neglect external disturbances entirely (see \emph{e.g.,}~\cite{nejati2022data,samari2024single,akbarzadeh2025formal}) or focus on bounded disturbances within a deterministic robust framework (see \emph{e.g.,}~\cite{bisoffi2022controller,BISOFFI20203953,LUPPI2024100914,akbarzadeh2024learning}). In contrast, the present work introduces a trajectory-based framework for analyzing the {probabilistic} safety of stochastic control systems subject to process noise with arbitrary distributions.

\textbf{Organization.} 
The remainder of the paper is structured as follows. Section \ref{Sec: Prob_Desc} introduces the notations and mathematical preliminaries, together with the definition of stochastic nonlinear polynomial systems. In the same section, the definitions of stochastic control barrier certificates and the associated safety guarantees are presented. In Section \ref{Physics-Guided}, the data-driven approach is proposed for synthesizing S-CBC and its associated safety controller solely based on noisy data. Section \ref{Diss} is dedicated to several discussions on the proposed approach, including a comparison between robust and stochastic analysis, which highlights the conservatism of robust methods in the existing literature compared to the probabilistic results presented in this work. Finally, simulation results are presented in Section \ref{Section: Simulation}, followed by the conclusion in Section \ref{Section: Conclusion}.

\section{Problem description}\label{Sec: Prob_Desc} 

\subsection{Notation}
We denote the set of real numbers by $\mathbb{R}$, while $\mathbb{R}_0^+$ and $\mathbb{R}^+$ represent the subsets of non-negative and positive real numbers, respectively. The set of non-negative and positive integers are, respectively, denoted by $\mathbb{N} = \{0,1,2,\dots\}$ and $\mathbb{N}^+ = \{1,2,\dots\}$. 
The notation $\mathbb{R}^n$ represents an $n$-dimensional Euclidean space, whereas $\mathbb{R}^{n \times m}$ denotes the space of real matrices with $n$ rows and $m$ columns. 
The identity matrix of size $n \times n$ is expressed as $\mathds I_{n}$.  We represent by \( \mathbb{1}_n \in \mathbb{R}^n \) the column vector of dimension \( n \) whose entries are all ones.
We denote by \( \| \cdot \| \) and \( \| \cdot \| _F\) the \emph{spectral} and \emph{Frobenius} norm of a matrix, respectively, while \( | \cdot | \) represents the \emph{Euclidean} norm of a vector or \emph{absolute value} of a scalar.
Given a collection of $N$ vectors $x_i \in \mathbb{R}^{n_i}$, where $n_i\in\mathbb{N}^+$ and $i \in \{1, \dots, N\}$, their concatenation into a single column vector is denoted by $x = [x_1; \dots; x_N]$, with a total dimension of $\sum_{i=1}^{N} n_i$.  Given $N$ vectors $x_i \in \R^n$, the matrix $x=[x_1 \, \, \ldots \,\, x_N]$ has dimensions $n \times N$. 
A {symmetric} matrix \(P \) is denoted positive definite by \( P\succ 0 \), and positive semi-definite by \(P \succeq 0 \). The transpose of \( P \) is represented as \(P^\top \)\!. In a {symmetric} matrix, $*$ represents the transposed entry corresponding to its symmetric counterpart. For symmetric matrices \( A, B \in \mathbb{R}^{n \times n} \), we write \( A \preceq B \) if \( B - A \) is positive semidefinite. The minimum and maximum eigenvalues of a symmetric matrix $P$ are denoted by $\lambda_{\min}(P)$ and $\lambda_{\max}(P)$, respectively. 
The trace of a matrix $P\in \mathbb{R}^{n \times n}$, denoted by $\operatorname{Tr}(P)$, is defined as the sum of its diagonal entries, \emph{i.e.,} $\operatorname{Tr}(P) = \sum_{i=1}^{n} P_{ii}.$
For a system $\Upsilon$ and a property $\Pi$, the notation $\Upsilon \models \Pi$ signifies that $\Upsilon$ satisfies $\Pi$. The empty set is denoted by $\emptyset$. For any square matrix \( P \in \mathbb{R}^{n \times n} \), we define \( P^2 := P P \). 

\subsection{Preliminaries}
This work considers the probability space $(\Omega, \mathbb{F}_\Omega, \mathbb{P}_\Omega)$, where $\Omega$ is the sample space, $\mathbb{F}_\Omega$ is a sigma-algebra on $\Omega$ comprising subsets of $\Omega$ as events, and $\mathbb{P}_\Omega$ is the probability measure that assigns probabilities to those events. It is assumed that random variables introduced in this paper are measurable functions of the form $X\!\!:(\Omega,\mathbb{F}_\Omega) \rightarrow (S_X,\mathbb{F}_X) $ such that any random variable $X$ induces a probability measure on its space $(S_X,\mathbb{F}_X)$. We directly present the probability measure on $(S_X,\mathbb{F}_X)$ without
explicitly mentioning the underlying probability space and the function $X$ itself. The topological space $S$ is a \emph{Borel} space if it is homeomorphic to a Borel subset of a Polish space, \textit{i.e.}, a separable and completely metrizable space. A Borel sigma algebra is denoted by $\mathbb{B}(S)$, and can be generated from any Borel space $S$. The map $f\!\!: S \rightarrow Y$ is measurable whenever it is Borel measurable. The covariance of a random variable $\varsigma \in \mathbb{R}^n$ is denoted by $\operatorname{Cov}(\varsigma) = \Sigma$. Given a collection of events \( \{\mathcal{E}_j\}_{j=1}^T \), their union, denoted by \( \bigcup_{j=1}^T \mathcal{E}_j \), implies that at least one of the \( \mathcal{E}_j \) occurs.  
The intersection, denoted by \( \bigcap_{j=1}^T \mathcal{E}_j \), indicates that all of the \( \mathcal{E}_j \) occur, simultaneously.

\subsection{Stochastic Nonlinear Polynomial Systems}
This work focuses on discrete-time stochastic nonlinear polynomial systems as the underlying model class, which is  formalized in the following definition.

\begin{definition}[\textbf{dt-SNPS}]\label{Def_1}
	A discrete-time stochastic nonlinear polynomial system (dt-SNPS) is defined as
	\begin{equation}\label{Eq_BBox}
		\Upsilon\!: x^+ = f(x) + g(x)u + \varsigma,  
	\end{equation}
	where:
	\begin{itemize}
		\item \(x\) denotes the state of the dt-SNPS, which is a random variable taking values in the Borel set $X\subseteq \mathbb R^n$; 
		\item $x^+$ represents the state variables at the next time step, \emph{i.e.,} $x^+ \coloneq x(k + 1), \; k \in \mathbb{N}$;
		\item  \(u \in U\) is the control input, with $U\subseteq \mathbb R^m$ being the Borel input space of dt-SNPS; 
		\item $\varsigma$ is a sequence of i.i.d. random variables, defined on the sample space $\Omega$, while taking values in a set $\mathcal{V}_{\varsigma}$, namely	 
		\begin{equation*}
			\varsigma:=\{\varsigma(k):\Omega\rightarrow \mathcal V_{\varsigma},\,\,k\in\N\};
		\end{equation*}
		\item
		$f\!:X \rightarrow X$ is a polynomial transition map;
		\item  \(g\!: X \rightarrow \mathbb{R}^{n \times m}\) is a matrix polynomial on $x$.
	\end{itemize}
\end{definition}
The dt-SNPS in (\ref{Eq_BBox}) can be equivalently expressed as
\begin{equation}\label{Eq_Llike}
	\Upsilon\!:  x^+ = A \mathcal{F}(x) + B{\mathcal{G}}(x)u + \varsigma,
\end{equation}
where $A \in \mathbb{R}^{n\times l}$ and $
B \in \mathbb{R}^{n\times q}$ are system and control matrices, while $\mathcal{F}(x)\in \mathbb{R}^{l}$, with $\mathcal{F}(\boldsymbol{0_n}) = \boldsymbol{0_l}$, and ${\mathcal{G}}(x) \in \mathbb{R}^{q\times m}$ are a vector of monomials and a matrix polynomial in $x$, respectively. Let $\mathbf{x}:\Omega \times \mathbb{N} \to X$ denote the random state sequence satisfying~\eqref{Eq_Llike}. This sequence is called the \textit{solution process} of $\Upsilon$ under the input trajectory $u(\cdot)$, starting from the initial state  $x(0) \in X$.

In this work, while both matrices $A$ and $B$ are considered \emph{unknown}, the following assumption is imposed on $\mathcal{F}(x)$ and $\mathcal{G}(x)$.
\begin{assumption}\label{Assum1}
It is assumed that \emph{extended} dictionaries (\emph{i.e.,} libraries or families of functions) associated with $\mathcal{F}(x)$ and $\mathcal{G}(x)$ are available\footnote{With a slight abuse of
notation, we use $\mathcal{F}(x)$ and $\mathcal{G}(x)$ interchangeably to denote
both the original and the extended dictionaries throughout the paper.}.
\end{assumption}
The dictionaries in Assumption~\ref{Assum1} are assumed to be sufficiently rich to capture all relevant components of the underlying dynamics, possibly at the expense of including additional terms (cf. the case studies). 

\begin{remark}[\textbf{On  Dictionaries $\mathcal{F}(x)$ and $\mathcal{G}(x)$}]\label{Remark: dictionary}
	Having an {extensive} dictionary for $\mathcal{F}(x)$ and $\mathcal{G}(x)$ is generally not a restriction. In many practical scenarios, particularly in electrical and mechanical engineering applications, system dynamics can often be derived from first principles. However, while such physical laws define the structural form of the dynamics (\emph{i.e.,} $\mathcal{F}(x)$ and $\mathcal{G}(x)$), the precise system parameters often remain unidentified (\emph{i.e.,} $A$ and $B$). This aligns well with our assumption that matrices $A$ and $B$ are fully unknown. Given the maximum degree of polynomials, one can consider all possible combinations of monomials up to that degree (cf. the Lorenz case study).
\end{remark}

Since $\mathcal{F}(\boldsymbol{0_n}) = \boldsymbol{0_l}$, without loss of generality, one can find a polynomial matrix $\mathcal{J}(x) \in \mathbb{R}^{l \times n}$, such that
\begin{align}\label{transform}
	\mathcal{F}(x) = \mathcal{J}(x) x.
\end{align}
This transformation facilitates expressing the proposed conditions in this work in terms of $\mathcal{J}(x)$, simplifying the computation (cf. \eqref{Eq_CBC3_theorem}). 

Our approach is \emph{distribution-free} and does not require exact knowledge of the noise parameters. Instead, it only assumes the availability of statistical information in the form of upper bounds, as stated in the following assumption.
\begin{assumption}\label{Assum2}
The mean $\mu$ and covariance $\Sigma$ are assumed to satisfy $\mu \mu^\top \preceq \Gamma_\mu$ and $\Sigma \preceq \Gamma_\Sigma$ for some known upper-bound matrices $\Gamma_\mu$ and $\Gamma_\Sigma$, respectively.
\end{assumption}
 A dt-SNPS with fully unknown system matrices $A$ and $B$, but access to extensive dictionaries  $\mathcal{F}(x)$ and $\mathcal{G}(x)$, along with upper bounds on the noise parameters, is referred to as a (partially) unknown system in our framework.

\begin{remark}[\textbf{On Estimation of Noise Parameters}]
~Estimating bounds on the noise parameters (\emph{i.e.,} $\Gamma_\mu$ and $\Gamma_\Sigma$) does not require direct observation of the process noise. Instead, such bounds can be inferred from repeated trajectory data collected under identical initial conditions and input sequences, where the variation across trajectories is solely induced by different noise realizations. In this setting, empirical estimates of the mean and covariance can be constructed from the observed next-state samples and subsequently aggregated over the data horizon. Conservative upper bounds $\Gamma_\mu$ and $\Gamma_\Sigma$ can then be obtained by appropriately inflating these estimates to account for statistical uncertainty. Hence, the proposed framework only requires upper bounds on the first two moments of the noise and does not rely on explicit measurements of the noise or knowledge of its underlying probability distribution.
\end{remark}

Since this work focuses on developing a safety certificate for the dt-SNPS in~\eqref{Eq_Llike}, we now formally define the corresponding safety specification.

\begin{definition}[\textbf{Safety Property}]\label{Safety} Consider a safety specification $\mathbb{S} = (X_\eta, X_\delta, \mathcal{T})$, where $X_\eta, X_\delta \subseteq X$ denote the prescribed \emph{initial} and \emph{unsafe} sets of the dt-SNPS, respectively. We assume $X_\eta \cap X_\delta = \emptyset$; otherwise, the system is considered unsafe from the outset with probability $1$. The system $\Upsilon$ is said to satisfy the safety specification within the finite time horizon $\mathcal{T} \in \mathbb{N}$, denoted by $\Upsilon \models \mathbb{S}$, if all trajectories initiated from $X_\eta$ avoid entering $X_\delta$ throughout the duration $\mathcal{T}$. Due to the stochastic nature of the system, the objective is to quantify the probability $\mathbb{P} \big\{\Upsilon\models \mathbb{S} \big\} \ge 1 - \beta_1$, where $\beta_1 \in (0,1]$. 
\end{definition}

The next subsection formally defines stochastic CBCs, which provides probabilistic safety certificates for the dt-SNPS, as outlined in Definition~\ref{Safety}.

\subsection{Stochastic control barrier certificates}

\begin{definition}[\textbf{S-CBC}~{\cite[Section 6]{lavaei2022automated}}] \label{S-CBC def}
	Consider a dt-SNPS $\Upsilon$, with $X_\eta,X_\delta \subseteq X $ being its {initial} and {unsafe} sets, respectively. A function $\mathcal{B}: X \to \mathbb{R}_0^+$ is called a stochastic control barrier certificate (S-CBC) for $\Upsilon$ if there exist $\eta, \delta \in \mathbb{R}^+$, with $\eta< \delta$, $\psi \in \mathbb{R}^+,$ and $\kappa \in (0,1)$, such that
	\begin{subequations}\label{Eq_CBC}
		\begin{align}
			\mathcal{B}(x) &\leq \eta,  \quad\quad\quad\quad \forall x \in X_{\eta}, \label{Eq_CBC1} \\
			\mathcal{B}(x) &\geq \delta, \quad\quad\quad\quad \forall x \in X_{\delta}, \label{Eq_CBC2} \end{align}
		and $\forall x \in X,\exists u \in U $ such that
		\begin{align}
			\mathbb{E}\Big[\mathcal{B}\big(A \mathcal{F}(x) + B{\mathcal{G}}(x)u + \varsigma\big) \,\,\big|\,\, x, u\Big] \leq \kappa\mathcal{B}(x) + \psi,\label{Eq_CBC3}
		\end{align}
	\end{subequations}
	where $	\mathbb{E}$ is the expected operator with respect to the process noise $\varsigma$. Accordingly, $u$ fulfilling (\ref{Eq_CBC3}) is a safety controller for the dt-SNPS. 
\end{definition}

The S-CBC $\mathcal{B}$ in Definition~\ref{S-CBC def}  plays a role analogous to a Lyapunov function, but for safety rather than stability. In particular, conditions~\eqref{Eq_CBC1} and~\eqref{Eq_CBC2} enforce a separation between the initial set $X_\eta$ and the unsafe set $X_\delta$ through suitable level sets of $\mathcal{B}$, while condition~\eqref{Eq_CBC3} ensures that the expected value of $\mathcal{B}$ along the system dynamics decreases up to an additive term $\psi$ capturing the effect of the process noise. We note that the initial and unsafe sets $X_\eta$ and $X_\delta$ are both given a priori. The objective is then to design a controller directly from data such that trajectories starting from $X_\eta$ avoid entering $X_\delta$ with quantified probabilistic guarantees. To achieve this, we present the following theorem, borrowed from the literature~\cite{1967stochastic,anand2022small}, to establish a lower bound on the probability that all trajectories of the dt-SNPS avoid the unsafe set within a finite time horizon.

\begin{theorem}[\textbf{Safety Guarantee}~\cite{1967stochastic,anand2022small}]\label{Lemma imp}
	Consider a dt-SNPS with an S-CBC $\mathcal{B}$ satisfying conditions \eqref{Eq_CBC}. The probability that all trajectories of dt-SNPS, denoted by $\mathbf{x}$, starting form an initial state $x_0 \in X_\eta$ under the controller $u(\cdot)$ do not reach an unsafe set $X_{\delta}$ within the time step $k\in [0,\mathcal T]$ is lower bounded by $\beta_1$, as
	\begin{equation}\label{Eq:9}
		\mathbb P \Big\{{\mathbf{x}}\notin{X}_\delta ~\text{for all}~ k\!\in\![0,\mathcal T] \,~ \big|~ x_0\Big\} \geq 1-\beta_1,
	\end{equation}
	where	
	\begin{equation*}
		\beta_1=  \begin{cases} 
			1-(1-\frac{\eta}{\delta})(1-\frac{\psi}{\delta})^{\mathcal T}\!, & \quad \quad \text{if } \delta \geq \frac{\psi}{{1-\kappa}}, \\
			(\frac{\eta}{\delta}){\kappa}^{\mathcal T}+(\frac{\psi}{{(1-{\kappa})}\delta})(1-\kappa^{\mathcal T}), & \quad \quad \text{if }\delta < \frac{\psi}{{1-\kappa}}.  \\
		\end{cases}
	\end{equation*}	
\end{theorem}\vspace{0.2cm}

The following proposition is a direct consequence of Theorem~\ref{Lemma imp}, where setting $\kappa$ in~\eqref{Eq_CBC3} to one yields a more relaxed condition while still providing the probabilistic safety guarantee in the sense of Definition~\ref{Safety}.

\begin{proposition}[\textbf{Safety Guarantee with Relaxation~\cite{salamati2024data}}]\label{Pro-1}
	If $\kappa$ in~\eqref{Eq_CBC3} is equal to one, the value of $\beta_1$ in~\eqref{Eq:9} is reduced to $\frac{\eta + \psi\mathcal T}{\delta}$.
\end{proposition}

\begin{remark}[\textbf{On $\kappa = 1$}]~
	The safety bound proposed in Proposition~\ref{Pro-1} is slightly looser than the one in~\eqref{Eq:9}, reflecting a trade-off between \emph{conservativeness} and \emph{tractability}. Specifically, by setting $\kappa = 1$ in condition~\eqref{Eq_CBC3}, the requirement for finding an S-CBC becomes less stringent, which potentially increases the feasibility of the associated condition. While this relaxation can improve the likelihood of finding a suitable S-CBC $\mathcal{B}$, it comes at the cost of a slightly weaker probabilistic guarantee compared to the tighter bound in~\eqref{Eq:9}.
\end{remark}

\begin{remark}[\textbf{On General Nonlinear Systems}]~
	The results of Theorem~\ref{Lemma imp}, along with its consequence in Proposition~\ref{Pro-1}, are applicable to a \emph{general class} of nonlinear systems in the model-based setting. However, this work specifically focuses on nonlinear \emph{polynomial} systems, as the data-driven conditions developed in the following sections are formulated using SOS optimization. This restriction enables tractable and certifiable data-driven analysis while still capturing a broad range of nonlinear behaviors.
\end{remark}

While the S-CBC defined in Definition~\ref{S-CBC def} effectively guarantees safety through the results of Theorem~\ref{Lemma imp}, its design remains intractable due to the presence of unknown system dynamics on the left-hand side of~\eqref{Eq_CBC3} (\emph{i.e.,} unknown matrices $A$ and $B$, and process noise $\varsigma$). Additionally, the expected operator in~\eqref{Eq_CBC3} further complicates the problem when the system dynamics are not explicitly known. In light of these key challenges, we now formally introduce the {data-driven} stochastic problem that this work aims to address.

\begin{resp}
\begin{problem} \label{P_safety}
	Consider a dt-SNPS of the form~\eqref{Eq_Llike} with unknown system matrices $A$ and $B$, and process noise $\varsigma$ with an arbitrary distribution. Given a safety specification $\mathbb{S} = (X_\eta, X_\delta, \mathcal{T})$, develop a data-driven stochastic approach, with confidence level $\beta_2 \in (0,1]$, that leverages input-state data collected under different noise realizations to design an S-CBC and its associated safety controller, while quantifying the probabilistic safety level $\beta_1 \in (0,1]$ as in~\eqref{Eq:9}, \emph{i.e.,}
	\begin{align}\label{Newff}
		\PP\Big\{\PP\big\{\Upsilon\vDash \mathbb{S}\big\}\geq 1- \beta_1\Big\}\geq 1-\beta_2.
	\end{align}\vspace{-0.5cm}
\end{problem}
\end{resp}

\begin{remark}[\textbf{Nested Probabilistic Guarantees}]~
	The probability expression in~\eqref{Newff} involves two distinct sources of randomness. The \emph{inner} probability $\beta_1$ is taken with respect to the process noise and reflects the intrinsic stochasticity of the dt-SNPS $\Upsilon$, which persists even in model-based settings (cf. Theorem~\ref{Lemma imp}). In contrast, the \emph{outer} probability $\beta_2$ is taken over the randomness of the collected trajectory data and captures the confidence level associated with the empirical approximation of \( \mathbb{E}[\varsigma \varsigma^\top] \) (cf. Lemma~\ref{lemma:empirical-moment-bound} and Theorem~\ref{T_final}).
\end{remark}

To address Problem~\ref{P_safety}, we propose our data-driven approach in the following section.

\section{Data-driven design of S-CBC and safety controller}\label{Physics-Guided}

In the proposed data-driven framework, data is collected from an experiment in the presence of {process noise}. To do so, we first raise the following assumption.

\begin{assumption}\label{Assum3}
All states of the system are assumed to be measurable, and the collected data consists of input-state trajectories of the system.
\end{assumption}

\begin{remark}[\textbf{On Input-State Data}]~
While the input-state assumption is common in data-driven approaches~\cite{de2019formulas,van2020noisy}, it may become restrictive in scenarios where not all states are measurable. Extending the present framework to input-output settings by leveraging notions such as uniform observability~\cite{dai2023data} constitutes a direction for future research.
\end{remark}

We initialize the system with a given state and apply a sequence of arbitrary inputs, measuring the resulting states generated by (\ref{Eq_Llike}) over time steps $k\in\{1,2,\dots,T\}$, with $T \in \mathbb{N}$ being the number of collected samples:
	\begin{subequations}\label{Eq_ST}
		\begin{align}
			\Xright^i&= [ x^i(1) \quad x^i(2) \quad\dots\quad x^i(T)],\label{Eq_STx} \\
			\mathbb{X}^i&= [ x(0) \quad ~ x^i(1) \quad\dots \quad x^i(T-1)],\label{Eq_STa} \\
			\mathbb{U} &= [ u(0) \quad ~ u(1) \quad ~ \dots \quad u(T-1)],\label{Eq_STb} \\
			\mathbb{Z}^i &= [ \varsigma^i(0) ~\quad \varsigma^i(1) \quad\dots \quad \varsigma^i(T-1)],\label{Eq_STd}
		\end{align}
	\end{subequations}
where $i \in\{1,2,\dots,N\}$ indicate trajectories with different noise realizations $\mathbb{Z}^i$, starting from the same initial condition $x(0)$ and following the same input sequence $\mathbb{U} $.

\begin{remark}[\textbf{On Input Data}]
	~During data collection, the input trajectory $u(\cdot)$ is kept fixed across all experiments, while the only source of variation between different trajectory realizations is the process noise. This setup ensures that the trajectories are comparable and the effect of the noise can be consistently estimated using the empirical mean (cf. Lemma~
	\ref{lemma:empirical-moment-bound}). Moreover, the framework assumes that the training data are gathered while maintaining system safety, which can be achieved, for instance, by applying inputs that do not drive the system toward unsafe regions or by validating the inputs in simulation prior to deployment.
\end{remark}

\begin{remark}[\textbf{On i.i.d. Noise vs. Non-i.i.d. Trajectories}]
	~It is important to distinguish between the statistical properties of the process noise and those of the resulting state trajectories. Specifically, the process noise $\{\varsigma(k)\}_{k=0}^{T-1}$ is assumed to be i.i.d., which is standard in stochastic control settings. However, despite this assumption, the resulting state trajectory $\{x(k)\}_{k=1}^{T}$ in~\eqref{Eq_ST} is inherently non-i.i.d., due to the recursive nature of the system dynamics: each state $x(k)$ depends on the entire history of prior states and noise realizations, introducing  dependence across the trajectory. This is in contrast to the scenario approach, which requires multiple statistically independent trajectories to ensure i.i.d. samples.
\end{remark}

Given the availability of an extended dictionaries for $\mathcal{F}(x)$ and $\mathcal{G}(x)$, the following trajectories can be extracted based on $\mathbb{X}^i$ and $\mathbb{U}$:
\begin{subequations}\label{Eq_mandq}
	\begin{gather} 
		\mathbb{F}^i \!=\! [ \mathcal{F}(x(0))\,\,\,\mathcal{F}(x^i (1))\,\dots\,\mathcal{F}(x^i (T\!-\!1))], \\
		{\mathbb{G}}^i  \!=\! [ \mathcal{G}(x(0)\!)u(0)\,\,\,\mathcal{G}(x^i (1)\!)u(1) \,\dots\,\mathcal{G}(x^i (T\!-\!1)\!)u(T\!-\!1))].
	\end{gather}
\end{subequations}

The proposed framework can be interpreted as consisting of two main steps throughout the remainder of the section. The first step focuses on constructing a probabilistic characterization of the noise effect using only statistical information, namely the mean and covariance bounds of the process noise. This step is developed through Lemmas~\ref{lemma:empirical-moment-bound} and~\ref{Lem1}, where Lemma~\ref{lemma:empirical-moment-bound}  quantifies the deviation between $\mathbb{E}[\varsigma\varsigma^\top]$ and its empirical approximation obtained from multiple trajectory realizations, and Lemma~\ref{Lem1} subsequently uses this result to construct the corresponding data-conformity (DC) set. The second step consists of synthesizing an S-CBC via S-lemma arguments, which is carried out in Theorem~\ref{T_final} (cf. condition~\eqref{Eq_CBC3_theorem}, which incorporates both the statistical noise information and the empirical average of trajectory data).

According to the first step discussed above, we aim at leveraging the \emph{second-moment} information of the noise, given by \( \mathbb{E}[\varsigma \varsigma^\top] \), to derive a valid upper bound for its empirical approximation \( \frac{1}{N}\sum_{i=1}^{N} \mathbb Z^{i}_{(\cdot)} \mathbb Z^{i\top}_{(\cdot)} \), where $(\cdot)$ indicates that the time step of noise realization $\mathbb{Z}^i$ can be arbitrarily chosen from $0$ to $T-1$. This bound will be crucial in the main result of this work (cf. Theorem~\ref{T_final}). We raise the following lemma that quantifies the norm distance between \( \mathbb{E}[\varsigma \varsigma^\top] \) and its empirical approximation \( \frac{1}{N}\sum_{i=1}^{N} \mathbb Z^{i}_{(\cdot)} \mathbb Z^{i\top}_{(\cdot)} \).

\begin{lemma}[\textbf{Formalizing Empirical Approximation}] \label{lemma:empirical-moment-bound}
Let $\epsilon>0$ be given. Consider \( N \) independent samples  \( \mathbb{Z}^1_{(\cdot)}, \dots, \mathbb{Z}^N_{(\cdot)} \in \mathbb{R}^n \) drawn from a Gaussian random vector \( \varsigma \in \mathbb{R}^n \), as in \eqref{Eq_STd}. Suppose Assumption~\ref{Assum2} holds with some known matrices $\Gamma_\mu$ and $\Gamma_\Sigma$. Then, one can quantify the distance between \( \mathbb{E}[\varsigma \varsigma^\top] \) and its empirical approximation \( \frac{1}{N}\sum_{i=1}^{N} \mathbb Z^{i}_{(\cdot)} \mathbb Z^{i\top}_{(\cdot)} \) to be within the threshold $\epsilon$ with the confidence of at least $1 - \bar\beta_2$, as
	\begin{align}\label{nd6sh}
		\mathbb{P} \Big( 
		\Big\| \frac{1}{N} \sum_{i=1}^N \mathbb{Z}^i_{(\cdot)} \mathbb{Z}^{i\top}_{(\cdot)} - \mathbb{E}[\varsigma \varsigma^\top] \Big\|_F 
		< \epsilon 
		\Big)
		\geq 
		1 -  \bar\beta_2,
	\end{align}
	with ${\bar\beta_2} = \frac{1}{N\epsilon^2} \big(  \operatorname{Tr}(\Gamma_\Sigma^2) + (\operatorname{Tr}(\Gamma_\Sigma))^2 + 2 \lambda_{\max}(\Gamma_\Sigma)\operatorname{Tr}(\Gamma_\mu) + 2\operatorname{Tr}(\Gamma_\Sigma)\operatorname{Tr}(\Gamma_\mu)\big).$
\end{lemma}

The proof of Lemma~\ref{lemma:empirical-moment-bound} in provided in the Appendix~\ref{app: proofs}.

\begin{remark}[\textbf{On Distribution-Free Characterization}]
	~While Lemma~\ref{lemma:empirical-moment-bound} is derived under a Gaussian assumption, ${\bar\beta_2}$ can be computed for any arbitrary distributions, as the concentration result itself is not inherently tied to Gaussianity. In particular, the use of Chebyshev’s and Markov’s inequalities renders the probabilistic bound distribution-free. The only step in the computation of ${\bar\beta_2}$ that depends on the specific distribution is the evaluation of the fourth moment $\mathbb{E}[\|\varsigma\|^4]$ (cf.~\eqref{fgsd2}). For any arbitrary distributions, this term can be replaced by an appropriate upper bound, provided the fourth moment exists.
\end{remark}

\begin{remark}[\textbf{On Choice of Frobenius Norm}]\label{Re1}
	~While the Frobenius norm is generally larger than the spectral norm for any given matrix (\emph{i.e.,} \( \|\cdot\| \leq \|\cdot\|_F \)), we deliberately adopt the Frobenius norm in Lemma \ref{lemma:empirical-moment-bound} due to the tighter probabilistic guarantees it affords under the proposed setting. Specifically, the Chebyshev-type bound under the Frobenius norm allows for an exact variance calculation in~\eqref{hsg|} as
	\[
	\mathbb{E}\left[ \left\| \varsigma \varsigma^\top - \mathbb{E}[\varsigma \varsigma^\top] \right\|_F^2 \right]
	= \mathbb{E}[\|\varsigma\|^4] - \left\| \mathbb{E}[\varsigma \varsigma^\top] \right\|_F^2\!,
	\]
	which leads to a clean and compact bound in terms of the mean and covariance. In contrast, the analogous spectral-norm bound relies on a looser inequality of the form  
	\[
	\mathbb{E}\left[ \left\| \varsigma \varsigma^\top - \mathbb{E}[\varsigma \varsigma^\top] \right\|^2 \right]
	\leq 2 \mathbb{E}[\|\varsigma\|^4] + 2 \left\| \mathbb{E}[\varsigma \varsigma^\top] \right\|^2\!\!,
	\]
	which introduces a conservative upper bound.
\end{remark}

\begin{remark}[\textbf{Tighter $\bar\beta_2$}]
	~ The bound offered by Lemma~\ref{lemma:empirical-moment-bound} is derived using Markov's inequality and Chebyshev-type arguments, which exhibit a polynomial decay in the confidence parameter. However, when the distribution of the noise vector $\varsigma \in \mathbb{R}^n$ satisfies stronger tail conditions, such as being sub-Gaussian or bounded, tighter results can be obtained using matrix Bernstein inequalities~\cite{tropp2015introduction}. These bounds leverage the boundedness or light-tailed nature of the noise and exhibit \emph{exponential} decay in the sample size $N$, in contrast to the Chebyshev-type bound where the confidence term $\bar{\beta}_2$ decays only as $1 / (N \epsilon^2)$.
\end{remark}

Building on this lemma, we now derive a valid upper bound on the empirical approximation \( \frac{1}{N}\sum_{i=1}^{N} \mathbb Z^{i}_{(\cdot)} \mathbb Z^{i\top}_{(\cdot)} \) by leveraging the statistical information  \( \mathbb{E}[\varsigma \varsigma^\top] \).

\begin{lemma}[\textbf{Data-Conformity Constraint}]\label{Lem1}
	Given an unknown dt-SNPS $\Upsilon$, suppose Assumption~\ref{Assum2} holds. Under Lemma~\ref{lemma:empirical-moment-bound}, one has
	\begin{align}\label{Eq_DC}
		\frac{1}{N} \sum_{i=1}^N \mathbb{Z}^i_{(\cdot)} \mathbb{Z}^{i\top}_{(\cdot)} \preceq \Gamma_\Sigma + \Gamma_\mu +\epsilon\mathds I_{n},
	\end{align}
	 with confidence at least $1 - \bar\beta_2$, where $\bar\beta_2$ is defined according to~\eqref{nd6sh}.
\end{lemma}

The proof of Lemma~\ref{Lem1} in provided in the Appendix~\ref{app: proofs}.

\begin{remark}[\textbf{Probabilistic vs. Robust Analysis}]~
		Unlike robust data-driven approaches that rely on \emph{worst-case} scenarios, the proposed framework leverages statistical information in the form of the mean and covariance of the noise. As a result, the method is applicable to noise distributions with \emph{unbounded} support and generally leads to less conservative conditions. This reduced conservatism, however, comes with a probabilistic characterization through the confidence parameter $\bar{\beta}_2$ in Lemma~\ref{Lem1}. Moreover, while the S-lemma is standard in control analysis, its use here is developed specifically for the stochastic data-driven setting through the construction of a data-conformity set based on empirical second-moment information obtained from multiple trajectory realizations.
\end{remark}

Having introduced the data-conformity constraint according to Lemma \ref{Lem1}, and following the second step discussed earlier in this section, we now proceed to propose our main data-driven result. We first specify the S-CBC and its controller as
\begin{align}\label{controller}
	\mathcal{B}(x) = x^\top  P x, \quad u = \mathcal K(x)x, 
\end{align}
where $P\succ 0$. By doing so, one can simplify the closed-loop form of system \eqref{Eq_Llike} as
\begin{align}\notag
	x^+ =\,&A \mathcal{F}(x) + B{\mathcal{G}}(x)u + \varsigma\\\notag
	\stackrel{(\ref{transform}),(\ref{controller})}{=}\,\!\!\!\!\!\!&  \,\,\,\,\,(A \mathcal{J}(x) + B{\mathcal{G}}(x)\mathcal K(x))x + \varsigma \\\label{Eq_CLs}
	=\,&  \Phi \Lambda(x)x + \varsigma,
\end{align}
with
\begin{align}\notag
	\Phi = [A~~~B], \quad\Lambda(x) = \left[\begin{array}{c} \mathcal{J}(x) \\
		\mathcal{G}(x)\mathcal  K(x)\end{array}\right]\!\!.
\end{align}
We are now ready to propose the main result of this work.

\begin{theorem}[\textbf{Data-Driven Design of S-CBC}] \label{T_final}
	Given an unknown dt-SNPS $\Upsilon$ as in Definition \ref{Def_1}, let Assumptions~\ref{Assum1}-\ref{Assum3} hold. Suppose there exist $\bar\eta, \bar\delta,  \in \mathbb{R}^+$, with $\bar\eta > \bar\delta$, $\kappa \in (0,1)$, matrices $\bar P\succ 0$ and $\bar {\mathcal  K}(x)$, and  $\alpha_{j=1,\dots,T}(x)\!\!:\mathbb{R}^n \to  \mathbb{R}_0^+$, such that
		\begin{subequations}\label{Eq_CBC_theorem}
		\begin{align}\label{Eq_CBC1_theorem}
			&\bar P - \bar\eta z_\eta z_\eta^\top ~\!\succeq 0,\\\notag
			&\text{with}~~X_\eta \subseteq \{x\in \mathbb{R}^n\!\!:~ xx^\top \preceq z_\eta z_\eta^\top\!, ~z_\eta \in \mathbb{R}^{n \times n}\},\\\label{Eq_CBC2_theorem}
			&\bar P  - \bar\delta z_\delta z_\delta^\top \preceq 0, \\\notag
			&\text{with}~~X_\delta \subseteq \{x\in \mathbb{R}^n\!\!: ~ xx^\top\succeq z_\delta z_\delta^\top\!, ~z_\delta \in \mathbb{R}^{n \times n}\},
			\end{align}
			\begin{align}\notag
			& 
			\begin{bmatrix} -\kappa \bar P  && 0  && 0\\ * &&  0 && \left[\begin{array}{c} \mathcal{J}(x)\bar P  \\
					\mathcal{G}(x)\bar{\mathcal  K}(x)\end{array}\right] \\
				* &&  * &&  -(1+\rho)^{-1}\bar P \end{bmatrix} \\\label{Eq_CBC3_theorem} 
			& ~~- \sum_{j=1}^{T}\alpha_j(x)\begin{bmatrix}\mathcal R^{DC_j} && 0  \\ * &&  0 \\
			\end{bmatrix}  \preceq 0,   \quad \quad \forall x \in X,
		\end{align}
	\end{subequations}
	with $\mathcal R^{DC_j}, j\in\{1,\dots,T\}$, as in \eqref{Eq_PI_new1} 
	for  some $\rho \in \mathbb{R}^+$. Then, $\mathcal{B}(x) = x^\top P x$, with $P = \bar P^{-1}$, is an S-CBC  for the  dt-SNPS and $u = \mathcal K(x)x$, with $\mathcal K(x) = \mathcal  {\bar K}(x) \bar P^{-1} =  \mathcal  {\bar K}(x)P$, is its corresponding safety controller with $\eta = \bar\eta^{-1}$, $\delta = \bar\delta^{-1}$ (where $\eta < \delta$), and
	\begin{align}\label{qqq}
		\psi = (1+\rho^{-1}) \operatorname{Tr}(P \Gamma_\mu) + \operatorname{Tr}(P\Gamma_\Sigma),
	\end{align}
	with a confidence of at least $1 - \beta_2$, where $\beta_2 = T \bar \beta_2$, with $\bar \beta_2$ as in \eqref{nd6sh}.
\end{theorem}

The proof of Theorem~\ref{T_final} in provided in the Appendix~\ref{app: proofs}.

\begin{remark}[\textbf{Barrier and Controller Structure}]~
		The quadratic S-CBC and the nonlinear state-feedback controller structure in~\eqref{controller} are mainly adopted for tractability and to derive the closed-form matrix inequality condition in~\eqref{Eq_CBC3_theorem} via the S-lemma. While higher-order polynomial barrier certificates can also be considered, they generally lead to significantly more complex SOS conditions. It is worth noting that the present framework is formulated as a feasibility problem, as is common in barrier certificate analysis. Nevertheless, in the stochastic setting, optimization objectives such as maximizing the achievable probabilistic safety guarantee in~\eqref{Eq:9} (equivalently, minimizing $\beta_1$) may also be pursued. This, however, introduces bilinear couplings among the design parameters, often necessitating that some of them be fixed \emph{a priori} to maintain computational tractability.
\end{remark}

\begin{figure*}[t!]
	\rule{\textwidth}{0.1pt}
	\begin{align}\label{Eq_PI_new1}
		\mathcal R^{DC_j} \!=\!\begin{bmatrix}
			\frac{1}{N} \sum_{i=1}^N\!\!\Xright^i_j  \Xright_j^{i\top} \!-\! \big(  \Gamma_\Sigma \!+\! \Gamma_\mu \!+\!\epsilon\mathds I_{n} \big) &&&  -\frac{1}{N} \sum_{i=1}^N\Xright^i_j \mathbb{H}_j^{i\top}\\
			* &&& ~~\frac{1}{N} \sum_{i=1}^N\mathbb{H}^i_j  \mathbb{H}_j^{i\top}\end{bmatrix}\!\!, \,\, \text{where }\,\, \mathbb{H}_j^i   \!=\! \left[\begin{array}{c} \mathbb{F}_j^i   \\
			{\mathbb{G}}_j^i   \end{array}\right]
	\end{align}
	\rule{\textwidth}{0.1pt}
\end{figure*}

We highlight that the key contribution of the paper lies in the development of a direct data-driven stochastic safety synthesis framework for unknown nonlinear systems using noisy trajectory realizations and S-CBC. Unlike most existing data-driven safety and control approaches based on worst-case robust formulations~\cite{de2019formulas,van2020noisy}, the proposed SOS-based framework directly constructs stochastic safety synthesis conditions from trajectory data through the proposed data-conformity formulation and explicitly incorporates empirical second-moment information of the process noise in~\eqref{Eq_CBC3_theorem}, thereby enabling probabilistic safety guarantees for noise distributions with potentially unbounded support.

Algorithm~\ref{alg} summarizes the required steps for the data-driven trajectory-based design of S-CBC and its corresponding safety controller.

\begin{remark}[\textbf{Computational Complexity Analysis}]~
	While SOS-based methods provide a systematic computational framework for polynomial control analysis and synthesis, it is well known that they may face scalability challenges for high-dimensional systems. In the proposed framework, the primary computational burden arises from solving the SOS condition in~\eqref{Eq_CBC3_theorem}, which is reformulated into a semidefinite program (SDP) by parsers such as SOSTOOLS. The overall computational complexity is primarily driven by the number of active decision variables. Let $r_1$ and $r_2$ denote the maximum polynomial degrees appearing in the polynomial controller matrix $\bar{\mathcal{K}}(x)$ and the multipliers $\alpha_j(x)$, respectively. The total number of explicit scalar decision variables, $v$, is given by
	\[
	v = mn\binom{n+r_1}{r_1}
	+
	T\binom{n+r_2}{r_2}
	+
	\frac{n(n+1)}{2}.
	\]
	When solved using standard primal-dual interior-point methods, both the \emph{time} and \emph{space} complexity generally scale polynomially with the number of decision variables $v$, with time scaling \emph{cubically} (\emph{i.e.,} $\mathcal{O}(v^3)$) and space scaling \emph{quadratically}  (\emph{i.e.,} $\mathcal{O}(v^2)$).
\end{remark}

\begin{algorithm}[t!] 
	\caption{Data-driven trajectory-based design of S-CBC and its safety controller with process noise}\label{alg}
	\begin{algorithmic}[1]\label{Algo}
		\REQUIRE  Safety specification $\mathbb{S} = (X_\eta, X_\delta, \mathcal{T})$, extended dictionaries $\mathcal{F}(x), \mathcal{G}(x)$, and bounds $\Gamma_\mu,\Gamma_\Sigma$
		\STATE Collect $\Xright^i,\mathbb{X}^i,$ and $ \mathbb{U}$, as in~\eqref{Eq_ST}, where $i \!\in\!\{1,2,\dots,N\}$\vspace{-0.4cm}
		\STATE Construct $\mathbb{F}^i$ and ${\mathbb{G}}^i$ as in \eqref{Eq_mandq}
		\STATE Initialize the desired $\epsilon$ and $N$, and compute $\bar\beta_2$ in Lemma~\ref{lemma:empirical-moment-bound}\vspace{-0.4cm}
		\STATE Initialize $\kappa \in (0,1)$ and $\rho \in \mathbb{R}^+$ in~\eqref{Eq_CBC_theorem} 
		\STATE Solve~\eqref{Eq_CBC_theorem} using \textsf{SeDuMi} and \textsf{SOSTOOLS}~\cite{prajna2004sostools} according to Theorem \ref{T_final} and compute $P = \bar P^{-1}$, $\mathcal K(x) = \mathcal  {\bar K}(x) \bar P^{-1} =  \mathcal  {\bar K}(x)P$, $\eta = \bar\eta^{-1}$, and $\delta = \bar\delta^{-1}$
		\STATE Compute $\psi$ according to \eqref{qqq}
		\STATE Quantify $\beta_1$ based on \eqref{Eq:9} and $\beta_2 = T \bar\beta_2$
		\ENSURE  S-CBC~$\mathcal{B}(x) = x^\top P x$, safety controller $u= \mathcal K(x) x$, and guaranteed probabilistic safety
		\begin{align*}
			\PP\Big\{\PP\big\{\Upsilon\vDash \mathbb{S}\big\}\geq 1- \beta_1\Big\}\geq 1-\beta_2
		\end{align*}\vspace{-0.4cm}
	\end{algorithmic}
\end{algorithm}

\section{Discussion}\label{Diss}

This section offers a broader perspective on the proposed approach, drawing on both the theoretical foundations established in this work and practical insights gained during its development and implementation.

{\bf Stochastic vs. Robust Analysis.} If one discards the noise distribution and instead performs robust analysis with respect to a bounded disturbance set, the resulting conditions can become overly conservative. Specifically, assume that \( \varsigma \) is a bounded disturbance with $\vert \varsigma\vert \leq \varkappa$, equivalently \( \varsigma \varsigma^\top \leq \varkappa^2 \mathds{I}_n \), without considering any probability distribution. Then, by applying the same data-conformity reasoning, the second term in the first entry of \( \mathcal{R}^{DC_j} \) in~\eqref{Eq_PI_new1} (\emph{i.e.,} $\Gamma_\Sigma \!+\! \Gamma_\mu \!+\!\epsilon\mathds I_{n}$) takes the form \( \varkappa^2 \mathds{I}_n \)~\cite{de2019formulas,van2020noisy}, rather than being expressed in terms of the noise parameters. This bound is larger than the quantity computed under distributional assumptions, making the matrix inequality in~\eqref{Eq_CBC3_theorem} more conservative and harder to satisfy. For instance, if $ \varsigma \in [-0.2, 0.2]^n$, then \( \varkappa^2 = 0.04n \). In contrast, if $\varsigma$ follows a uniform distribution as \( \varsigma \sim \mathcal{U}(\underbrace{-0.2}_a, \underbrace{0.2}_b)^n \), the second moment becomes 
\begin{align*}
\mathbb{E}[\varsigma \varsigma^\top] = \Sigma + \mu \mu^\top &= \frac{(b-a)^2}{12} \mathds{I}_n   + (\frac{a+b}{2})^2 \mathbb{1}_n \mathbb{1}_n^\top\\
&= \frac{0.16}{12} \mathds{I}_n \approx 0.0133 \mathds{I}_n,
\end{align*}
which is, even for $n=1$, approximately \emph{one-third} of the worst-case bound \( 0.04 \)   (\emph{i.e.,} $66\%$ less conservative element-wise), thereby facilitating the satisfaction of the matrix inequality in~\eqref{Eq_CBC3_theorem}. In the case study section, it is demonstrated that a robust approach fails to yield a safety controller for a physical example under bounded disturbances due to the conservatism inherent in worst-case analysis. In contrast, the proposed framework successfully synthesizes a safety controller that satisfies the safety constraints with high probability.

Similar benefits hold when assuming other noise distributions, \emph{e.g.,} normal or exponential. These distributions not only lead to \emph{smaller} expected second moments than conservative worst-case bounds, but they also offer a more realistic representation of real-world noise characteristics, primarily due to their \emph{unbounded} support.

\begin{remark}[\textbf{Conservatism of Moment Bounds}]\label{new_rem}~
	The conservatism of the proposed framework is directly influenced by the upper bounds $\Gamma_\mu$ and $\Gamma_\Sigma$ on the mean and covariance of the noise. In particular, larger values of these bounds make the matrix inequality condition in~\eqref{Eq_CBC3_theorem}  more restrictive and therefore reduce the feasibility region of the synthesis problem. Moreover, since the quantity $\Psi$ in~\eqref{qqq} also depends on $\Gamma_\mu$ and $\Gamma_\Sigma$, conservative moment bounds may decrease the achievable safety probability (\emph{i.e.,} $1-\beta_1$) in~\eqref{Eq:9}. They also reduce the level of confidence (\emph{i.e.,} $1-\beta_2$), as they appear in $\bar\beta_2$ in~\eqref{nd6sh}. Therefore, tighter knowledge of the noise moments generally leads to less conservative probabilistic safety guarantees.
\end{remark}

{\bf Trade-off Between $\epsilon$ and $\bar\beta_2$.} Lemma \ref{lemma:empirical-moment-bound} establishes a closed-form relationship between $\bar\beta_2$, $\epsilon$, and the number of empirical realizations $N$. As indicated by the proposed bound, for a fixed sample size $N$, the probability term $\bar\beta_2$ decreases as the tolerance parameter $\epsilon$ increases. This illustrates a fundamental trade-off: demanding higher accuracy (\emph{i.e.,} smaller $\epsilon$) reduces the confidence level ($1 - \bar\beta_2$), whereas accepting a larger deviation allows for a stronger probabilistic guarantee. This closed-form relation also gives a lower bound for $N$ based on $\epsilon$ and $\bar\beta_2$ as
\begin{align*}
N \geq& \frac{1}{\bar\beta_2\epsilon^2}  \big(  \operatorname{Tr}(\Gamma_\Sigma^2) + (\operatorname{Tr}(\Gamma_\Sigma))^2 + 2 \lambda_{\max}(\Gamma_\Sigma)\operatorname{Tr}(\Gamma_\mu) \\
& + 2\operatorname{Tr}(\Gamma_\Sigma)\operatorname{Tr}(\Gamma_\mu)\big).
\end{align*}
We note that $\bar{\beta}_2$ is intended to satisfy $0<\bar{\beta}_2<1$. If the resulting bound exceeds $1$ (\emph{e.g.,} due to insufficient data or overly conservative moment bounds), the concentration inequality becomes vacuous and no meaningful probabilistic guarantee can be inferred. As expected from standard concentration-bound analysis, $\bar{\beta}_2$ decreases as the number of trajectory realizations increases.

\begin{remark}[\textbf{On Infinitely Many Trajectories}]~
	As the number of trajectory realizations tends to infinity, \emph{i.e.,} $N\to\infty$, the confidence bound $\bar{\beta}_2$ in Lemma~ \ref{lemma:empirical-moment-bound} converges to zero, and the empirical approximation of $\mathbb{E}[\varsigma\varsigma^\top]$ converges to its true value with confidence approaching one. Consequently, the outer probabilistic layer associated with data uncertainty vanishes asymptotically. However, the probabilistic safety level $\beta_1$ in~\eqref{Eq:9} remains unchanged, as it reflects the intrinsic stochasticity induced by the process noise and is independent of the amount of available data.
\end{remark}

{\bf Multiple Noise Realizations.} While $N$ trajectories should be collected in \eqref{Eq_ST} to capture the stochastic behavior of the dt-SNPS, the proposed matrix inequality condition in \eqref{Eq_CBC3_theorem} needs to be solved only once. This is achieved by leveraging the empirical average of the collected trajectories, as embedded in the data-conformity matrix $\mathcal{R}^{DC_j}$ in \eqref{Eq_PI_new1}. As a result, the {stochastic nature} of the system is captured through data, without the need to solve multiple matrix inequalities. This design significantly reduces computational overhead, as the matrix inequality condition in \eqref{Eq_CBC3_theorem} is evaluated just once, and the averaging step does not introduce additional complexity in the optimization process.

{\bf Infinite-Horizon Guarantees.} If the stochasticity in~\eqref{Eq_BBox} is multiplicative \emph{i.e.,} $ \varsigma  \odot x$, with $\odot$ being the \emph{Hadamard product} (element-wise multiplication), then the constant $\psi$ in~\eqref{Eq_CBC3} could be zero. Accordingly, the S-CBC $\mathcal B$ satisfying condition~\eqref{Eq_CBC3} with $\psi = 0$ is \emph{non-negative supermartingale}~\cite[Chapter I]{1967stochastic}. In this case, one can employ Definition~\ref{S-CBC def} with $\psi = 0$ and provide an upper bound on the probability that all trajectories of dt-SNPS do not reach unsafe regions within \emph{infinite} time horizons as
	\begin{equation}\label{Eq:10}
		\mathbb P \Big\{{\mathbf{x}}\notin{X}_\delta ~\text{for all}~ k\!\in\![0, \infty)  \,~ \big|~ x_0\Big\} \geq 1-\beta_1,
	\end{equation}
	where $\beta_1 = \frac{\eta}{\delta}$. We note that requiring an S-CBC $\mathcal B$ with $\psi = 0$ is potentially more restrictive than the conditions in Theorem~\ref{Lemma imp}, but offers the advantage of enabling probabilistic guarantees over {infinite} time horizons.
	
	\begin{table*}[t!]
		\centering
		\caption{Overview of data-driven results across three stochastic control systems with unknown dynamics.  For each system, runtime for collecting data (\texttt{RT1}), runtime for solving conditions~\eqref{Eq_CBC_theorem} (\texttt{RT2}), and memory usage (\texttt{MU}) are reported.\label{tab:benchmark}}
		\begin{tabular}{@{}ccccccccccccc@{}}
				\toprule
				\cmidrule(lr){3-5} \cmidrule(lr){6-7} \vspace{-0.25cm}\\
				{System} & $N$ & $T$ & $\mathcal T$  & $\epsilon$  & $\Gamma_\mu$ & $\Gamma_\Sigma$&  $\beta_1$ & $\bar\beta_2$& $\beta_2$& \makecell{\texttt{RT1}\\(sec)}& \makecell{\texttt{RT2}\\(sec)}& \makecell{\texttt{MU}\\(Mb)}\\
				\midrule
				\myalign{l}{Lorenz} & $77$   & $10$  & $100$   & $0.1$  & $ \boldsymbol{0}_3 $ & $ 0.006\, \mathds{I}_3$ & $0.08$  & $5 \times 10^{-4}$  & $0.005$  & $0.01$ & $1.58$ & $13.54$ \\
				\midrule
				\myalign{l}{Chen}  & $328$  & $7$  &   $100$ &  $0.2$ & $\boldsymbol{0}_3$ & $ 0.008\, \mathds {I}_3$ &  $0.05$ & $5 \times 10^{-5}$  &  $0.00035$ & $0.01$ & $0.74$ & $8.84$ \\
				\midrule
				\myalign{l}{Spacecraft} & $1283$  & $8$  &   $20$ &  $0.01$ & $\boldsymbol{0}_3$ & $0.0075\, \mathds {I}_3$ &  $0.07$ & $0.005$  &  $0.04$ & $0.03$ & $1.37$ & $12.73$ \\
				\bottomrule
		\end{tabular}
	\end{table*}

{\bf Stability Analysis.}
While the data-driven results of this work are presented in the context of {safety} analysis for stochastic control systems, the proposed approach can be utilized to assess other essential system properties, such as \emph{mean-square stability}, as its analysis typically relies on a condition involving the expected value of a Lyapunov function, closely resembling the structure of the condition in \eqref{Eq_CBC3}.

{\bf Limitations.} Similar to any approach, our findings have also some limitations that are worth mentioning. The primary issue lies in the class of nonlinear systems considered, which is restricted to polynomial dynamics. While many real-world engineering systems can be effectively modeled by polynomial nonlinearities (cf. benchmark case studies), extending the approach to handle more general nonlinear systems (\emph{e.g.,} involving $\sin, \cos$) is a direction for future work. It is worth noting that this restriction stems from the SOS nature of our condition~\eqref{Eq_CBC3_theorem}, which requires the system dynamics to be polynomial. Another limitation concerns the choice of the S-CBC, which is \emph{quadratic} in our current setting (\emph{i.e.} $\mathcal{B}(x) = x^\top  P x$). While a quadratic S-CBC is sufficient to ensure safety in many nonlinear systems, the proposed setting can be generalized to higher-order polynomials in future work, for instance using $\mathcal{B}(x) = \mathcal{F}(x)^\top P \mathcal{F}(x)$.

\section{Simulation Results}\label{Section: Simulation}
In this section, we demonstrate the effectiveness of our data-driven stochastic approach through its application on three benchmarks, including chaotic Lorenz-type systems (called Lorenz and Chen)~\cite{strogatz2018nonlinear} and a spacecraft~\cite{khalil2002control}. A brief yet general overview of these case studies is provided in Table \ref{tab:benchmark}. All simulations are performed in \MATLAB \textsl{R2023b} on a MacBook Pro (Apple M2 Pro - 32GB memory). 

The primary objective across all benchmarks is to design an S-CBC and its corresponding safety controller for systems with an unknown mathematical model, \emph{i.e.,} unknown $A$ and $B$, and an arbitrary noise distribution. To achieve this, under Algorithm~\ref{alg}, we collect input-state trajectories under different noise realizations and construct an S-CBC and its associated safety controller that satisfy the proposed conditions in~\eqref{Eq_CBC_theorem}. Figures~\ref{fig:4} and~\ref{fig:5} illustrate the closed-loop trajectories of the Lorenz and Chen systems under the synthesized controller. Details of the underlying models, and designed matrices and controllers are provided in the Appendix~\ref{app: simulations}.

As shown in Table~\ref{tab:benchmark}, the Chen system requires a higher number of noise realizations ($N = 328$) than the Lorenz system. Nevertheless, the runtime required to collect this data (\texttt{RT1}) is $0.01$ seconds. In addition, the runtime for solving the conditions in~\eqref{Eq_CBC_theorem} (\texttt{RT2}) is $0.74$ seconds, demonstrating the high scalability of our approach. This computation time is lower than that of the Lorenz system due to the shorter horizon length ($T = 7$). In contrast to the Lorenz and Chen systems, the spacecraft system uses a smaller $\epsilon = 0.01$, which results in a higher required number of samples ($N = 1013$). However, this facilitates the satisfaction of condition~\eqref{Eq_CBC3_theorem}, as $\epsilon$ appears in the data-conformity matrix $\mathcal{R}^{DC_j}$ in~\eqref{Eq_PI_new1}.

\begin{figure*}[!t]
	\centering
	\begin{subfigure}[t]{0.32\textwidth}
		\includegraphics[width=\linewidth]{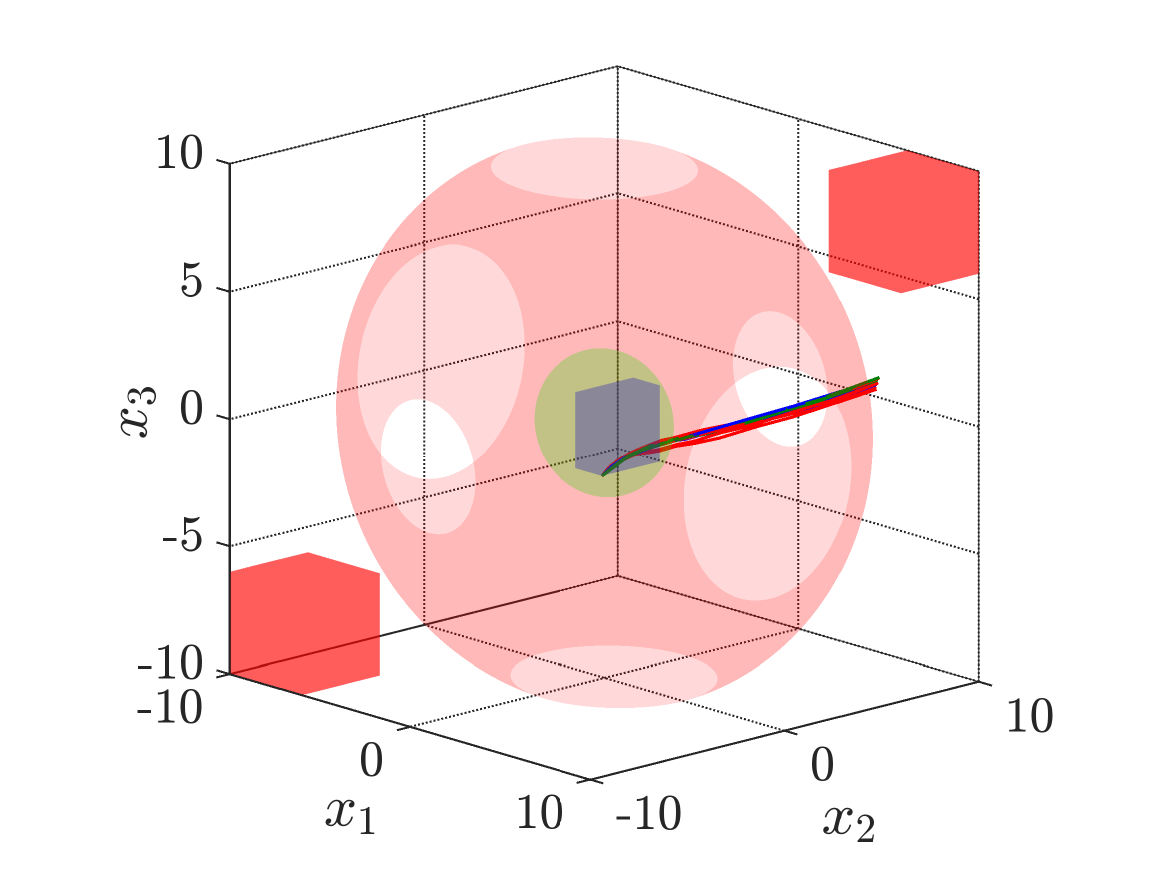}
	\end{subfigure}
	\hspace{-0.01cm}
	\begin{subfigure}[t]{0.32\textwidth}
		\includegraphics[width=\linewidth]{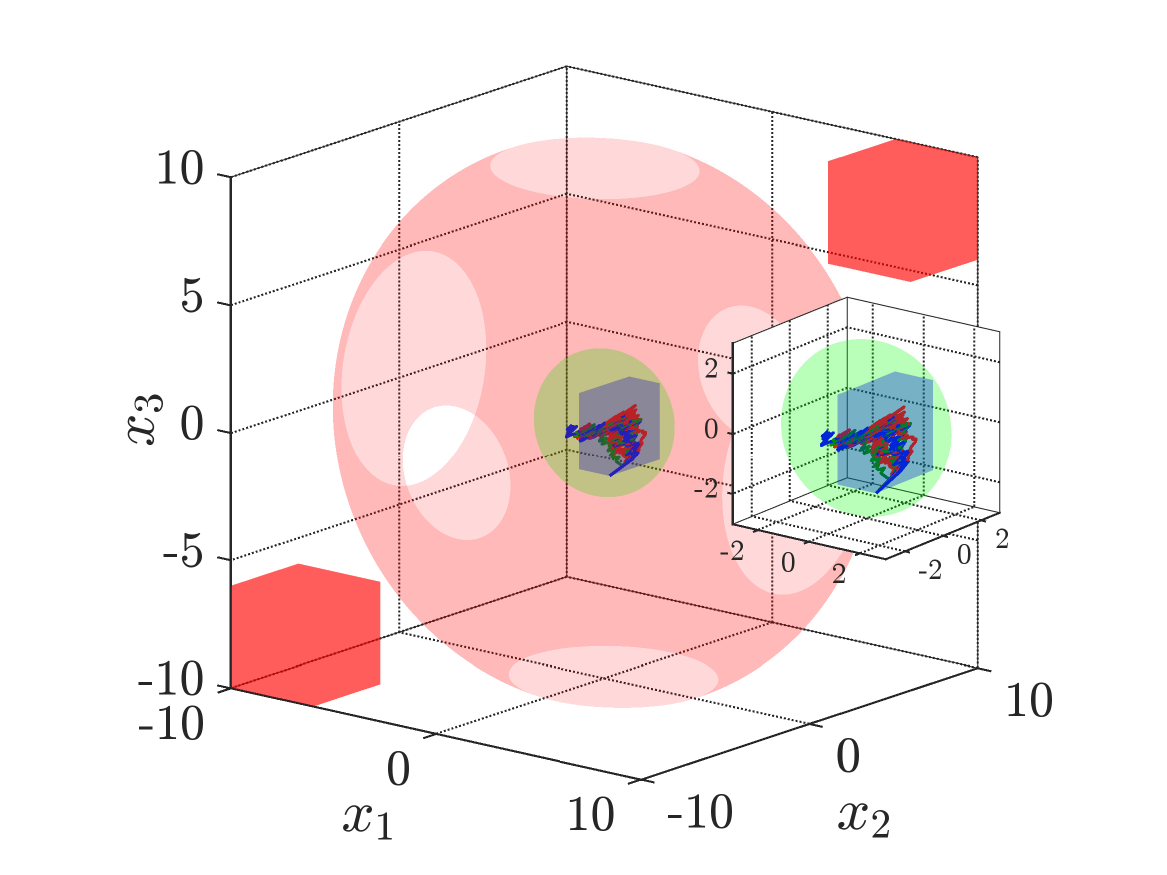}
	\end{subfigure}
	\hspace{-0.01cm}
	\begin{subfigure}[t]{0.32\textwidth}
		\includegraphics[width=\linewidth]{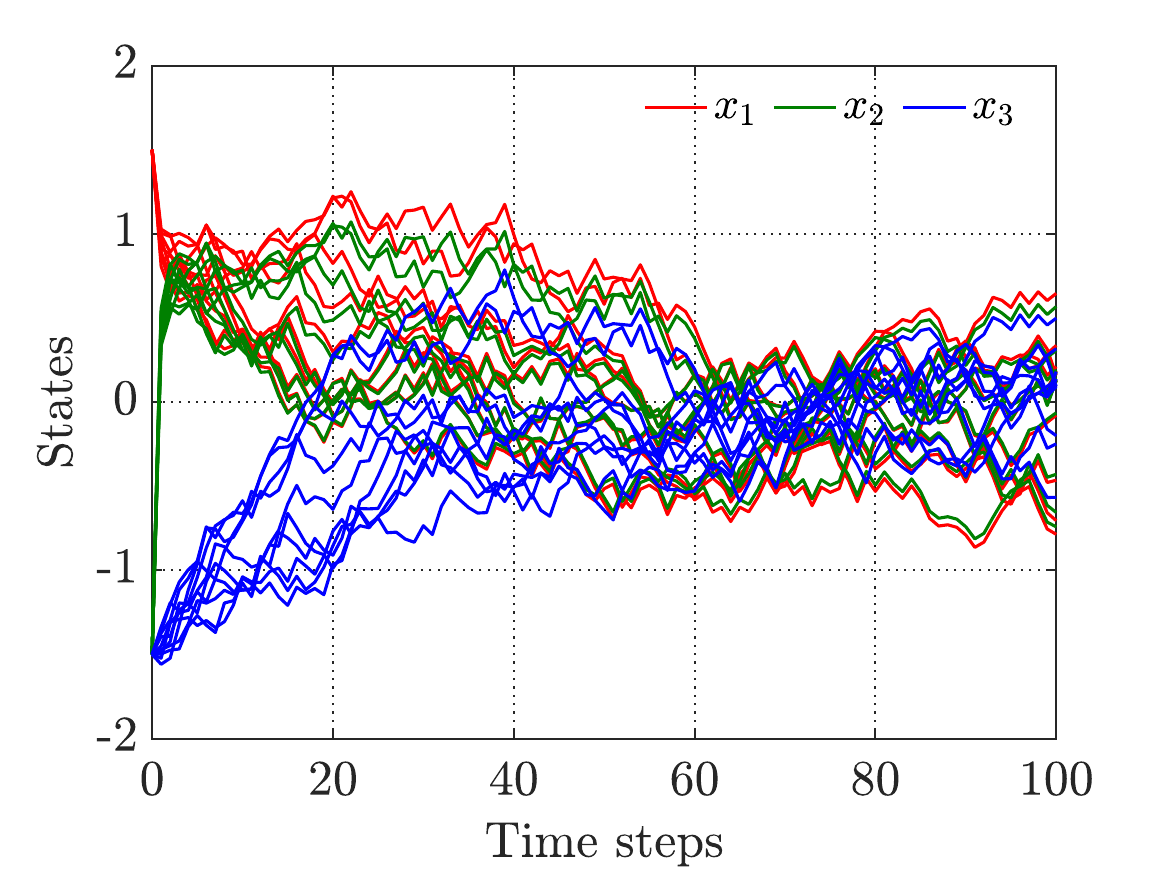}
	\end{subfigure}
	
	\caption{{\bf Lorenz:}  Closed-loop state trajectories under an ad hoc controller~(\textbf{left figure}), which violates the safety specification, and the designed data-driven controller in \eqref{sdhg}~(\textbf{middle figure}), shown together with a magnified inset. Initial and unsafe regions are depicted by purple \protect\bluesquare\ and red  \protect\redsquare\ boxes, while $\mathcal {B}(x) = \eta$ and $\mathcal {B}(x) = \delta$ are indicated by~\protect\greensquare\ and~\protect\pinksquare\,, respectively. The \textbf{right figure} indicates arbitrary trajectories under different noise realizations over a time horizon of $\mathcal{T}=100$, consistent with our safety guarantee, demonstrating the robustness of the framework.
	}
	\label{fig:4}
\end{figure*}

As another observation, while the probabilistic safety guarantees (\emph{i.e.,} $1 - \beta_1$) for the Lorenz and Chen systems are $92\%$ and $95\%$, respectively, all closed-loop trajectories under the designed controller, as shown in Figs.~\ref{fig:4} and~\ref{fig:5}, satisfy the safety specification. This indicates that the data-driven controller is somewhat conservative, which is expected due to the use of the S-CBC as a Lyapunov-like function. However, this conservatism comes with the benefit of providing formal safety guarantees, rather than relying solely on empirical safety validation. It is worth noting that smaller values of $\Gamma_\mu$ and $\Gamma_\Sigma$ generally lead to smaller values of $\beta_1$ and $\beta_2$, thereby improving both the probabilistic safety guarantee and the associated confidence level (cf. Remark~\ref{new_rem}). For instance, in the Lorenz example, if $\Gamma_\mu$ and $\Gamma_\Sigma$ are conservatively estimated as $\Gamma_\mu = 0.001 \mathds{I}3$ and $\Gamma_\Sigma = 0.01\mathds{I}3$, the corresponding value of $\beta_1$ increases to $0.1$, reducing the probabilistic safety guarantee from $92\%$ (reported in Table~\ref{tab:benchmark} as $1-\beta_1$) to $90\%$. Moreover, the corresponding value of $\beta_2$ increases to $0.01$, reducing the confidence level from $99.5\%$ (reported in Table~\ref{tab:benchmark} as $1-\beta_2$) to $99\%$. 

As the final simulation study, we consider the spacecraft system to compare our stochastic analysis with the robust methods proposed in the literature~\cite{de2019formulas,van2020noisy}. We consider the case where the disturbance $\varsigma$ is bounded, \emph{i.e.,} $|\varsigma| \leq \varkappa$, or equivalently $\varsigma \varsigma^\top \leq \varkappa^2 \mathds{I}_3$, without assuming any specific probability distribution. Under this assumption, and following the same data-conformity reasoning, the second term in the first entry of $\mathcal{R}^{DC_j}$ in~\eqref{Eq_PI_new1} is replaced by $\varkappa^2 \mathds{I}_3$~\cite{de2019formulas,van2020noisy}, rather than being expressed in terms of noise distribution parameters. By considering the disturbance bound as $[-0.15,   0.15]^3$, yielding $\varkappa^2 = 0.0225$, no feasible controller can be obtained using the robust analysis.

We then allow the disturbance to follow a uniform distribution and aim at solving the problem using our proposed stochastic framework. In this case, the problem is readily solved using the parameters reported in Table~\ref{tab:benchmark}, as the term $\Gamma_\Sigma + \Gamma_\mu + \epsilon \mathds{I}_3 = 0.0175\mathds{I}_3$ in $\mathcal{R}^{DC_j}$~\eqref{Eq_PI_new1} is approximately \emph{$22\%$ smaller in each entry} than the worst-case bound $\varkappa^2\mathds{I}_3 = 0.0225\mathds{I}_3$. This case study clearly demonstrates that if the robust worst-case approach fails to synthesize a safety controller under bounded disturbances due to its conservatism, our data-driven stochastic framework is capable of synthesizing a safety controller that satisfies the safety constraints with potentially high probability. This highlights that, despite introducing some probabilistic risk, stochastic analysis can be more effective and less conservative than robust analysis in many scenarios. For the sake of a fair comparison, it should be noted that when robust analysis is feasible, it can be more advantageous, as it ensures safety without introducing any risk to the system.

\begin{figure*}[!t]
	\centering
	\begin{subfigure}[t]{0.32\textwidth}
		\includegraphics[width=\linewidth]{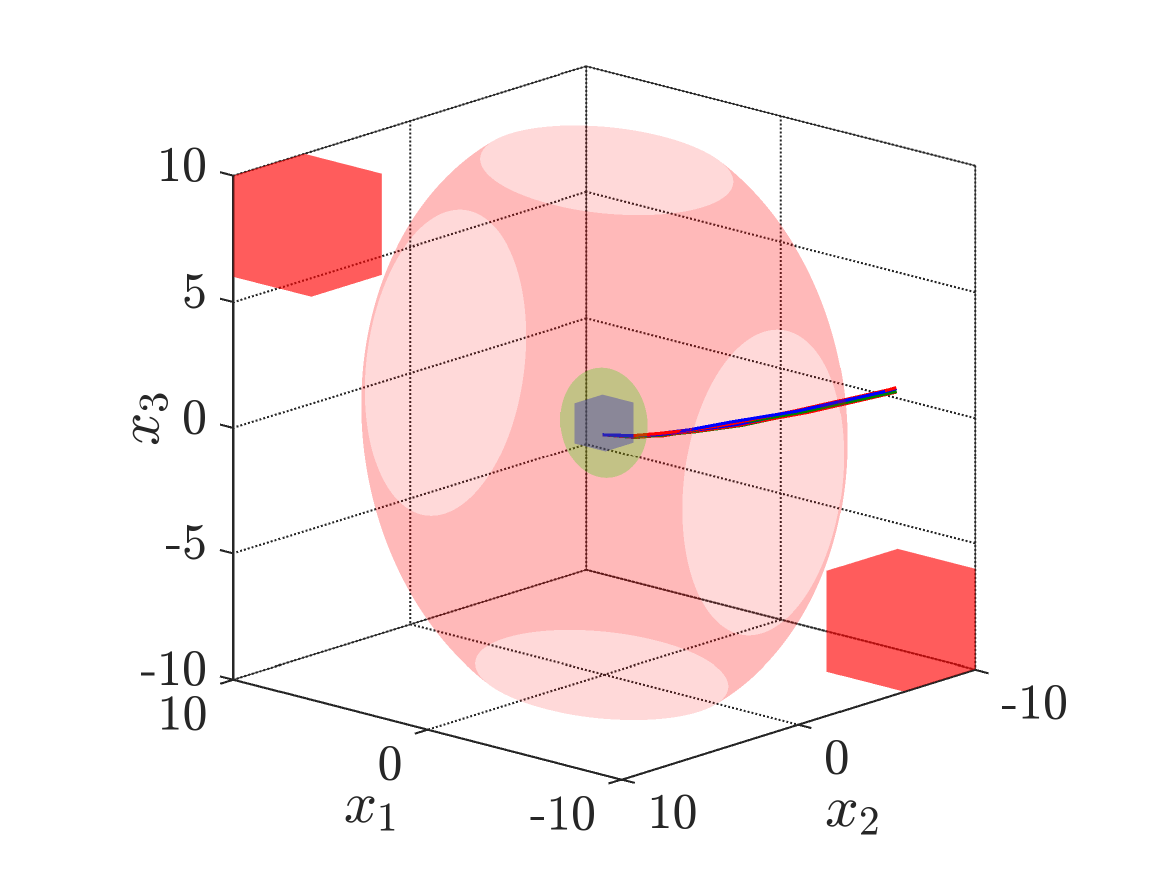}
	\end{subfigure}
	\hspace{-0.01cm}
	\begin{subfigure}[t]{0.32\textwidth}
		\includegraphics[width=\linewidth]{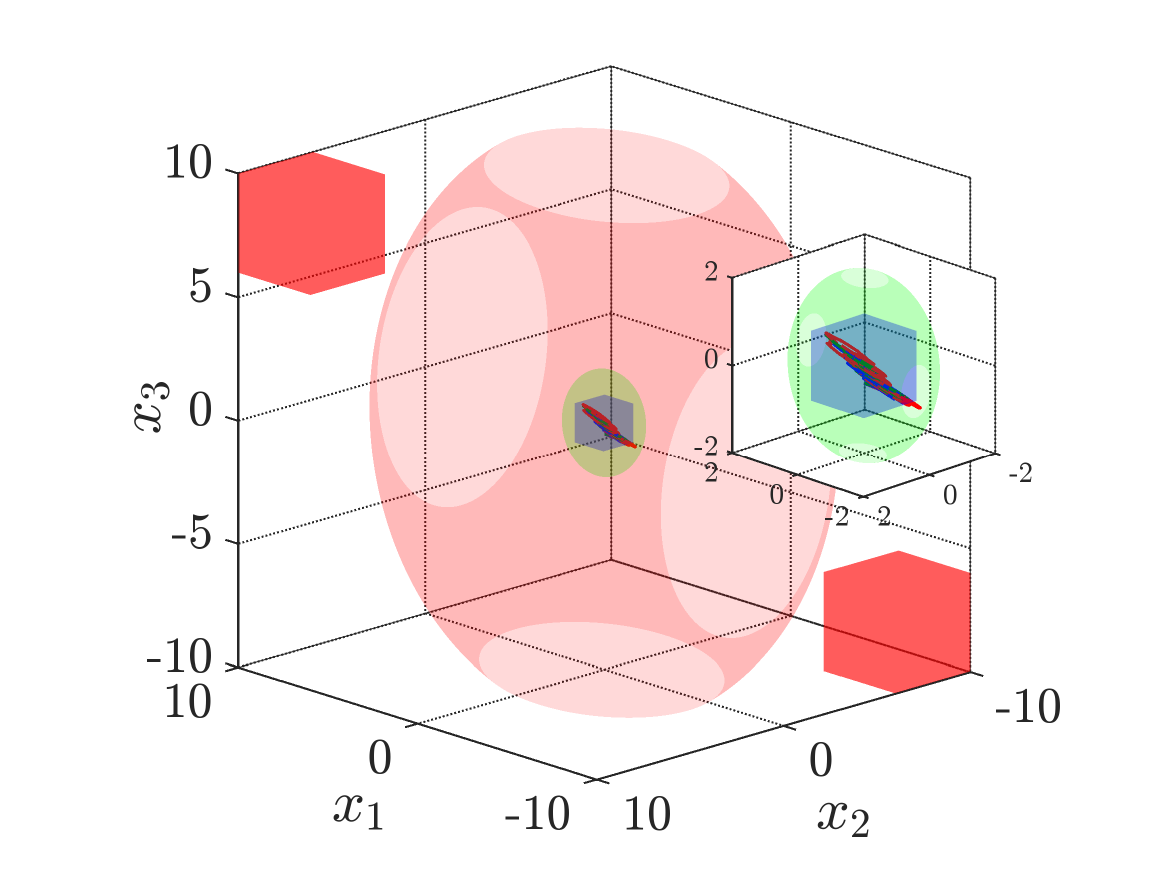}
	\end{subfigure}
	\hspace{-0.01cm}
	\begin{subfigure}[t]{0.32\textwidth}
		\includegraphics[width=\linewidth]{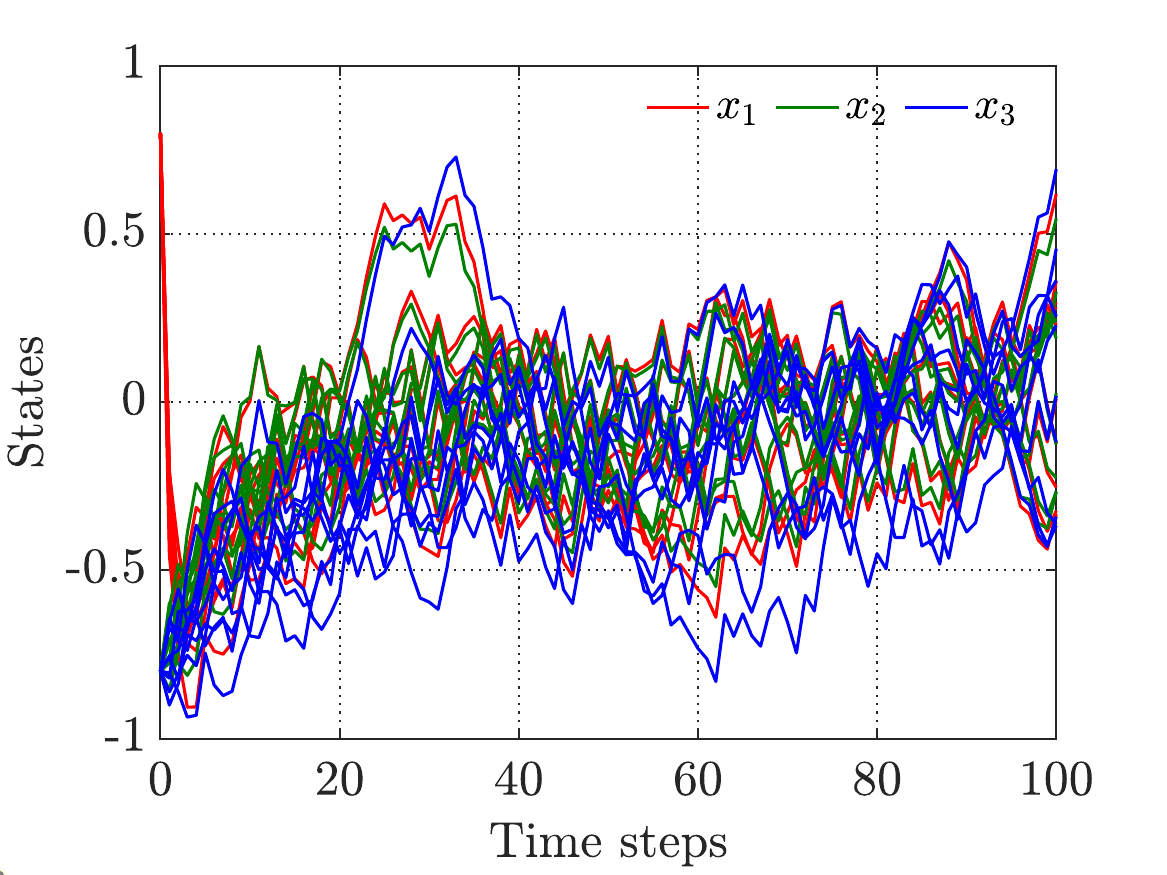}
		
	\end{subfigure}
	\caption{{\bf Chen:}  Closed-loop state trajectories under an ad hoc controller~(\textbf{left figure}), which violates the safety specification, and the designed data-driven controller in \eqref{hdjfh}~(\textbf{middle figure}) with magnified inset. Initial and unsafe regions are depicted by purple \protect\bluesquare\ and red  \protect\redsquare\ boxes, while $\mathcal {B}(x) = \eta$ and $\mathcal {B}(x) = \delta$ are indicated by~\protect\greensquare\ and~\protect\pinksquare\,, respectively. The \textbf{right figure} indicates arbitrary trajectories under different noise realizations over a time horizon of $\mathcal{T}=100$.}
	\label{fig:5}
\end{figure*}

\section{Conclusion}\label{Section: Conclusion}
This work developed a non-i.i.d. trajectory-based  framework for {stochastic} control systems with both {unknown} dynamics and arbitrary noise distributions. Prior work in this area has often considered dynamical systems subject to {bounded} disturbances, enabling robust analysis. Despite their appeal, such assumptions may not hold in real-world settings, and the corresponding {worst-case} analysis typically introduces significant conservatism. In contrast, this work addressed systems influenced by {process noise}, introducing a stochastic data-driven scheme which is capable of quantifying {probabilistic} safety guarantees for unknown models with a certified confidence level. The method relies on finite-horizon data collection from multiple noise realizations and employs stochastic control barrier certificates to derive probabilistic safety guarantees over finite horizons. The resulting conditions are cast as an SOS program, using only empirical average of trajectory data  and the statistical properties of the process noise. Ongoing work explores extensions to more general class of nonlinear stochastic systems with dynamics beyond polynomials.

\section{acknowledgment}
The author would like to thank Omid Akbarzadeh and MohammadHossein Ashoori for their assistance with the simulations in Section~\ref{Section: Simulation}.

\bibliographystyle{IEEEtran}
\bibliography{biblio}	

\section{Appendix: Proof of Statements}\label{app: proofs}

{\bf Proof of Lemma~\ref{lemma:empirical-moment-bound}.}
Consider the non-negative random variable
\[
\Big\| \frac{1}{N} \sum_{i=1}^N \mathbb{Z}^i_{(\cdot)} \mathbb{Z}^{i\top}_{(\cdot)} - \mathbb{E}[\varsigma \varsigma^\top] \Big\|_F^2\!.
\]
By employing the Chebyshev's inequality~\cite{hernandez2001chebyshev} for matrix-valued random variables under the Frobenius norm as well as applying the Markov's inequality~\cite{billingsley1995probability}, we have
\begin{align}\notag
	\mathbb{P} &\Big( \Big\| \frac{1}{N} \sum_{i=1}^N \mathbb{Z}^i_{(\cdot)} \mathbb{Z}^{i\top}_{(\cdot)} - \mathbb{E}[\varsigma \varsigma^\top] \Big\|_F \geq \epsilon \Big)\\\notag
	&= \mathbb{P} \Big( \Big\| \frac{1}{N} \sum_{i=1}^N \mathbb{Z}^i_{(\cdot)} \mathbb{Z}^{i\top}_{(\cdot)} - \mathbb{E}[\varsigma \varsigma^\top] \Big\|_F^2 \geq \epsilon^2 \Big)\\\label{lll}
	&\leq
	\frac{1}{\epsilon^2} \mathbb{E} \Big[ \Big\| \frac{1}{N} \sum_{i=1}^N \mathbb{Z}^i_{(\cdot)} \mathbb{Z}^{i\top}_{(\cdot)} - \mathbb{E}[\varsigma \varsigma^\top] \Big\|_F^2 \Big]\!.
\end{align}
Now, since the samples \( \mathbb{Z}^i_{(\cdot)} \) are i.i.d. realizations of the random vector $\varsigma$, one has
\begin{align}\notag
	\mathbb{E} & \Big[ \Big\| \frac{1}{N} \sum_{i=1}^N \mathbb{Z}^i_{(\cdot)} \mathbb{Z}^{i\top}_{(\cdot)} - \mathbb{E}[\varsigma \varsigma^\top] \Big\|_F^2 \Big]\\\notag
	&=
	\frac{1}{N^2} \sum_{i=1}^N \mathbb{E} \left[ \left\| \mathbb{Z}^i_{(\cdot)} \mathbb{Z}^{i\top}_{(\cdot)} - \mathbb{E}[\varsigma \varsigma^\top] \right\|_F^2 \right]\\\label{llll}
	&=
	\frac{1}{N} \mathbb{E} \Big[ \Big\| \varsigma \varsigma^\top - \mathbb{E}[\varsigma \varsigma^\top] \Big\|_F^2 \Big]\!.
\end{align}
Combining~\eqref{lll} and \eqref{llll} yields
\begin{align}\notag
	\mathbb{P} &\Big( \Big\| \frac{1}{N} \sum_{i=1}^N \mathbb{Z}^i_{(\cdot)} \mathbb{Z}^{i\top}_{(\cdot)} - \mathbb{E}[\varsigma \varsigma^\top] \Big\|_F \geq \epsilon \Big)\\\notag
	&\leq\frac{1}{N \epsilon^2} \mathbb{E} \left[ \left\| \varsigma \varsigma^\top - \mathbb{E}[\varsigma \varsigma^\top] \right\|_F^2 \right]\!\!.
\end{align}
We now aim to compute an upper bound for the variance term
\[
\mathbb{E} \left[ \left\| \varsigma \varsigma^\top - \mathbb{E}[\varsigma \varsigma^\top] \right\|_F^2 \right]
= \mathbb{E} \Big[ \mathrm{Tr} \big( ( \varsigma \varsigma^\top - \mathbb{E}[\varsigma \varsigma^\top])^2 \big) \Big].
\]
By expanding the square inside the trace, one has
\begin{align}\notag
	\mathbb{E} &\Big[ \mathrm{Tr} \big(( \varsigma \varsigma^\top - \mathbb{E}[\varsigma \varsigma^\top])^2 \big) \Big]\\\notag
	& = \mathbb{E} \big[ \mathrm{Tr} \left( \varsigma \varsigma^\top \varsigma \varsigma^\top \right) \big]
	- 2 \, \mathrm{Tr} \left( \mathbb{E}[\varsigma \varsigma^\top]^2 \right)
	+ \mathrm{Tr} \left( \mathbb{E}[\varsigma \varsigma^\top]^2 \right) \\\label{hsg|}
	&= \mathbb{E}[\| \varsigma \|^4] - \left\| \mathbb{E}[\varsigma \varsigma^\top] \right\|_F^2.
\end{align}
To proceed with the proof, one can invoke the well-known moment identities, for Gaussian random variables~\cite{petersen2008matrix}, as
\begin{subequations}
	\begin{align}\notag
		\mathbb{E}[\| \varsigma \|^4] &= 2\| \Sigma \|_F^2 + (\operatorname{Tr}(\Sigma))^2 + 4 \mu^\top \Sigma \mu \\\label{fgsd2}
		&~~~+ 2\operatorname{Tr}(\Sigma) \| \mu \|^2 + \| \mu \|^4, \\\label{fgsd3}
		\left\| \mathbb{E}[\varsigma \varsigma^\top] \right\|_F^2 &= \left\| \Sigma + \mu \mu^\top \right\|_F^2 = \| \Sigma \|_F^2 + 2 \mu^\top \Sigma \mu + \| \mu \|^4.
	\end{align}
\end{subequations}
By substituting \eqref{fgsd2}-\eqref{fgsd3} into the variance expression in~\eqref{hsg|} and by utilizing the \emph{cyclic} property of the trace, we have
\begin{align*}
	&\mathbb{E} \left[ \left\| \varsigma \varsigma^\top \!\!-\! \mathbb{E}[\varsigma \varsigma^\top] \right\|_F^2 \right]
	= 2\| \Sigma \|_F^2 + (\operatorname{Tr}(\Sigma))^2 + 4 \mu^\top \Sigma \mu  \\
	&~~~+ 2\operatorname{Tr}(\Sigma) \| \mu \|^2 + \| \mu \|^4 - \left( \| \Sigma \|_F^2 + 2 \mu^\top \Sigma \mu + \| \mu \|^4 \right) \\
	&= \| \Sigma \|_F^2  + (\operatorname{Tr}(\Sigma))^2 + 2 \mu^\top \Sigma \mu + 2\operatorname{Tr}(\Sigma) \| \mu \|^2  \\
	&= \operatorname{Tr}(\Sigma^2) \!+\!  (\operatorname{Tr}(\Sigma))^2 \!+\!  2 \mu^\top \Sigma \mu \!+\!  2\operatorname{Tr}(\Sigma) \| \mu \|^2  \\
	&\leq \operatorname{Tr}(\Sigma^2) \!+\!  (\operatorname{Tr}(\Sigma))^2 \!+\!  2 \lambda_{\max}(\Sigma)\Vert \mu\Vert^2 \!+\!  2\operatorname{Tr}(\Sigma) \| \mu \|^2 \\
	&=  \operatorname{Tr}(\Sigma^2) \!+\!  (\operatorname{Tr}(\Sigma))^2 \!+\!  2 \lambda_{\max}(\Sigma)\operatorname{Tr}(\mu\mu^\top)\!+\!  2\operatorname{Tr}(\Sigma)\operatorname{Tr}(\mu\mu^\top\!).
\end{align*}
Given that  $\mu \mu^\top \preceq \Gamma_\mu$ and $\Sigma \preceq \Gamma_\Sigma$, it follows that $\operatorname{Tr}(\mu\mu^\top) \leq \operatorname{Tr}(\Gamma_\mu)$, $\lambda_{\max}(\Sigma) \leq \lambda_{\max}(\Gamma_\Sigma)$, and $\operatorname{Tr}(\Sigma^2) \leq \operatorname{Tr}(\Gamma_\Sigma^2)$. Then one has
\begin{align}\notag
	\mathbb{E} &\left[ \left\| \varsigma \varsigma^\top - \mathbb{E}[\varsigma \varsigma^\top] \right\|_F^2 \right]\leq \operatorname{Tr}(\Gamma_\Sigma^2) + (\operatorname{Tr}(\Gamma_\Sigma))^2  \\\label{bvs}
	&+ 2 \lambda_{\max}(\Gamma_\Sigma)\operatorname{Tr}(\Gamma_\mu) + 2\operatorname{Tr}(\Gamma_\Sigma)\operatorname{Tr}(\Gamma_\mu).
\end{align}
Combining \eqref{lll}, \eqref{llll}, and \eqref{bvs}, we obtain
\begin{align}\notag
	\mathbb{P} &\Big( \Big\| \frac{1}{N} \sum_{i=1}^N \mathbb{Z}^i_{(\cdot)} \mathbb{Z}^{i\top}_{(\cdot)} - \mathbb{E}[\varsigma \varsigma^\top] \Big\|_F \geq \epsilon \Big)\\\notag
	&\leq
	\frac{1}{N\epsilon^2} \big(  \operatorname{Tr}(\Gamma_\Sigma^2) + (\operatorname{Tr}(\Gamma_\Sigma))^2 \\
	&~~~+ 2 \lambda_{\max}(\Gamma_\Sigma)\operatorname{Tr}(\Gamma_\mu) + 2\operatorname{Tr}(\Gamma_\Sigma)\operatorname{Tr}(\Gamma_\mu)\big).
\end{align}
Taking the complement gives the result as
\begin{align*}
	\mathbb{P}& \Big( 
	\Big\| \frac{1}{N} \sum_{i=1}^N \mathbb{Z}_{(\cdot)}^i \mathbb{Z}_{(\cdot)}^{i\top} 
	- \mathbb{E}[\varsigma \varsigma^\top] \Big\|_F
	< \epsilon 
	\Big)\geq
	1 - {\bar\beta_2},
\end{align*}
with 
\begin{align*}
	{\bar\beta_2} =& \frac{1}{N\epsilon^2} \big(  \operatorname{Tr}(\Gamma_\Sigma^2) + (\operatorname{Tr}(\Gamma_\Sigma))^2 + 2 \lambda_{\max}(\Gamma_\Sigma)\operatorname{Tr}(\Gamma_\mu) \\
	& + 2\operatorname{Tr}(\Gamma_\Sigma)\operatorname{Tr}(\Gamma_\mu)\big),
\end{align*}
which completes the proof. \hfill $\blacksquare$

{\bf Proof of Lemma~\ref{Lem1}.}
Let $\varsigma = [\varsigma_1; \varsigma_2; \dots; \varsigma_n]$. It is clear that
\begin{align}\label{nks1}
	\mathbb{E}\Big[\varsigma\varsigma^\top\Big] = \operatorname{Cov}(\varsigma) + \mathbb{E}\Big[\varsigma\Big]\mathbb{E}\Big[\varsigma^\top\Big]\!.
\end{align}
Given that the covariance of $\varsigma$, \emph{i.e. $\operatorname{Cov}(\varsigma)$}, is $\Sigma$  and its mean is $\mu$, and since  $\mu \mu^\top \preceq \Gamma_\mu$ and $\Sigma \preceq\Gamma_\Sigma$, one has
\begin{align}\label{New17}
	\mathbb{E}\Big[\varsigma\varsigma^\top\Big] = \Sigma + \mu\mu^\top \preceq\Gamma_\Sigma+ \Gamma_\mu.
\end{align}
Given the results of Lemma~\ref{lemma:empirical-moment-bound} and since $\Vert \cdot\Vert \leq \Vert \cdot \Vert_F$, the inequality
\begin{align}\label{xxx}
	\Big\| &\frac{1}{N} \sum_{i=1}^N \mathbb{Z}^i_{(\cdot)} \mathbb{Z}^{i\top}_{(\cdot)} - \mathbb{E}[\varsigma \varsigma^\top] \Big\| \\\notag
	&\leq \Big\| \frac{1}{N} \sum_{i=1}^N \mathbb{Z}^i_{(\cdot)} \mathbb{Z}^{i\top}_{(\cdot)} - \mathbb{E}[\varsigma \varsigma^\top] \Big\|_F 
	< \epsilon 
\end{align}
holds true with a confidence of at least $1 - \bar\beta_2$.
Since both $\frac{1}{N} \sum_{i=1}^N \mathbb{Z}^i_{(\cdot)} \mathbb{Z}^{i\top}_{(\cdot)}$ and $\mathbb{E}[\varsigma \varsigma^\top]$ are symmetric and the bound in ~\eqref{xxx} is expressed in the spectral norm, one has
\begin{align}\label{New651}
	\frac{1}{N} \sum_{i=1}^N \mathbb{Z}^i _{(\cdot)}\mathbb{Z}^{i\top}_{(\cdot)} - \epsilon \mathbb I_n \prec	\mathbb{E}\Big[\varsigma\varsigma^\top\Big]
\end{align}
which holds true with a confidence of at least $1 - \bar\beta_2$.
Given inequalities \eqref{New17} and \eqref{New651}, we have
\begin{align}\label{New65}
	\frac{1}{N} \sum_{i=1}^N \mathbb{Z}^i_{(\cdot)} \mathbb{Z}^{i\top}_{(\cdot)} \preceq \Gamma_\Sigma + \Gamma_\mu +\epsilon\mathds I_{n}.
\end{align} 
Since inequality~\eqref{New17} follows deterministically from Assumption~\ref{Assum2}, it holds for every realization. Moreover, according to Lemma~\ref{lemma:empirical-moment-bound}, inequality~\eqref{New651} holds with confidence at least $1-\bar\beta_2$. Therefore, inequality~\eqref{New65} also holds with confidence at least $1-\bar\beta_2$, which completes the proof.  \hfill $\blacksquare$

{\bf Proof of Theorem~\ref{T_final}.}
	We first show that condition \eqref{Eq_CBC3_theorem}  ensures the satisfaction of condition \eqref{Eq_CBC3}. Since $\mathcal{B}(x) = x^\top P x$, we have
	\begin{align} \notag
		&\mathbb{E}\Big[\mathcal{B}(A \mathcal{F}(x) + B{\mathcal{G}}(x)u + \varsigma) \,\,\big|\,\, x, u\Big]\\\notag &\stackrel{(\ref{Eq_CLs})}{=}\,\mathbb{E}\Big[\big(\Phi \Lambda(x)x + \varsigma\big)^\top P\big(\Phi \Lambda(x)x + \varsigma\big)\,\,\big|\,\, x, u\Big]\\\notag
		& \,\,= \big[\Phi \Lambda(x)x\big]^\top \!P \big[\Phi \Lambda(x)x\big] \!+\! 2\underbrace{\big[\Phi \Lambda(x)x\big]^\top \sqrt{P}}_{a}\underbrace{\sqrt{P}\mathbb{E}\big[\varsigma\,\,\big|\,\, x, u\big]}_{b}\\\label{Eq_barriernext} 
		&~~~ + \mathbb{E}\big[\varsigma^\top P \varsigma\,\,\big|\,\, x, u\big].  
	\end{align}
	According to the Cauchy-Schwarz inequality, \emph{i.e.,}  $a b \leq \vert a \vert \vert b \vert,$ for any $a^\top, b \in \R^{n}$, followed by
	employing Young's inequality \cite{caverly2019lmi}, \emph{i.e.,} $\vert a \vert \vert b \vert \leq \frac{\rho}{2} \vert a \vert^2 + \frac{1}{2\rho} \vert b \vert^2$, for any $\rho \in \mathbb{R}^+$, one has
	\begin{align}\notag
		\mathbb{E}&\Big[\mathcal{B}(A \mathcal{F}(x) + B{\mathcal{G}}(x)u + \varsigma) \,\,\big|\,\, x, u\Big]\\\notag
		&~ \leq (1+\rho)\big[\Phi \Lambda(x)x\big]^\top P \big[\Phi \Lambda(x)x\big] \\ \label{Eq_nextBB}
		&~~~+ \frac{1}{\rho} \, \mathbb{E}\big[\varsigma \mid x,u\big]^\top P \, \mathbb{E}\big[\varsigma \mid x,u\big] + \mathbb{E}\big[\varsigma^\top P \varsigma\,\,\big|\,\, x, u\big]. 
	\end{align}
	Given that $\varsigma^\top P \varsigma$ is scalar, and by employing the \emph{cyclic} property of the trace, one has
	\[
	\mathbb{E}\Big[\varsigma^\top P \varsigma\Big] \!=\! \mathbb{E}\Big[\operatorname{Tr}(\varsigma^\top P \varsigma)\Big] \!=\! \mathbb{E}\Big[\operatorname{Tr}(P \varsigma \varsigma^\top)\Big] \!=\! \operatorname{Tr}(P \mathbb{E}\Big[\varsigma \varsigma^\top\Big]).
	\]
	Given that $\mathbb{E}\Big[\varsigma\varsigma^\top\Big] = \Sigma + \mu\mu^\top$ according to \eqref{New17}, we have
	\begin{align*}
		\mathbb{E}\Big[\varsigma^\top P \varsigma\Big] \!=\! \operatorname{Tr}\big(P (\Sigma + \mu\mu^\top)\big) = \operatorname{Tr}(P \Sigma) + \operatorname{Tr}\big(P \mu\mu^\top).
	\end{align*}
	Since \( P\succ 0 \) and  $\mu \mu^\top \preceq \Gamma_\mu$, $\Sigma \preceq \Gamma_\Sigma$, one has
	\begin{align}\label{New87}
		\mathbb{E}\Big[\varsigma^\top P \varsigma\Big] \leq \operatorname{Tr}(P\Gamma_\Sigma) + \operatorname{Tr}\big(P\Gamma_\mu).
	\end{align}
	Similarly, one has 
	\begin{align}\notag
		\mathbb{E}\big[\varsigma \big]^\top P \, \mathbb{E}\big[\varsigma \big] &= \operatorname{Tr}(\mathbb{E}\big[\varsigma\big]^\top P \, \mathbb{E}\big[\varsigma\big])= \operatorname{Tr}(P \, \mathbb{E}\big[\varsigma\big]\mathbb{E}\big[\varsigma \big]^\top) \\\label{nhsd}
		&= \operatorname{Tr}(P \mu \mu ^\top)\!\leq\! \operatorname{Tr}(P \Gamma_\mu).
	\end{align}
	By applying \eqref{New87} and \eqref{nhsd} to \eqref{Eq_nextBB}, one has 
	\begin{align}\notag
		\mathbb{E}&\Big[\mathcal{B}(A \mathcal{F}(x) + B{\mathcal{G}}(x)u + \varsigma) \,\,\big|\,\, x, u\Big]\\\notag
		&\leq (1+\rho)\big[\Phi \Lambda(x)x\big]^\top P \big[\Phi \Lambda(x)x\big]\\\label{Eq_nextBB1}
		&~~~ + \underbrace{(1+\rho^{-1})\operatorname{Tr}(P \Gamma_\mu)+ \text{Tr}(P\Gamma_\Sigma)}_\psi. 
	\end{align}
	By subtracting $\kappa\mathcal{B}(x)$ from both sides of \eqref{Eq_nextBB1}, we have
	\begin{align} \notag
		\mathbb{E}&\Big[\mathcal{B}(A \mathcal{F}(x) + B{\mathcal{G}}(x)u + \varsigma) \,\,\big|\,\, x, u\Big] - \kappa\mathcal{B}(x)\\\notag
		&\leq  (1+\rho)\big[\Phi \Lambda(x)x\big]^\top P \big[\Phi \Lambda(x)x\big] - \kappa\mathcal{B}(x) + \psi. 
	\end{align}
	By defining 
	\begin{align}\notag  
		\mathbb L(x) := (1+\rho)\big[\Phi \Lambda(x)x\big]^\top P \big[\Phi \Lambda(x)x\big] - \kappa \overbrace{x^\top P x}^{\mathcal B(x)},
	\end{align}
	it is evident that if $\mathbb L(x)\leq 0$, then 
	\begin{align} \notag
		\mathbb{E}\Big[\mathcal{B}(A \mathcal{F}(x) + B{\mathcal{G}}(x)u + \varsigma) \,\,\big|\,\, x, u\Big] - \kappa\mathcal{B}(x) \leq \psi.
	\end{align}
	Our attention now turns to ensuring that the constraint $\mathbb{L}(x) \leq 0$ is met. The expression $\mathbb{L}(x)$ can be expanded as 
	\begin{align*}
		\mathbb L(x) =\,& (1+\rho)x^\top \Lambda(x)^\top \Phi^\top  P\Phi \Lambda(x)x - \kappa x^\top P x \notag \\
		=\,& x^\top \big((1+\rho)\Lambda(x)^\top \Phi ^\top P\Phi \Lambda(x)-\kappa P\big) x. 
	\end{align*}
	To enforce $\mathbb L(x)\leq 0$, it is sufficient to satisfy
	\begin{align}\label{new09}
		(1+\rho)\Lambda(x)^\top \Phi ^\top P\Phi \Lambda(x)- \kappa P \preceq 0. 
	\end{align}
	Using Schur complement, \eqref{new09} is equivalent to
	\begin{align} 
		(1+\rho)\Phi \Lambda(x) P^{-1}\Lambda(x)^\top \Phi ^\top- \kappa P^{-1} \preceq 0, \notag
	\end{align}
	which could be recast in the following quadratic form:
	\begin{align}\label{Eq_CBC3_new}
		\mathcal{Q}^{S\text{-}CBC}_{\Phi}(x):=\thetaTildeI^\top \mathcal R^{S\text{-}CBC}(x) \thetaTildeI \preceq 0,
	\end{align}
	with
	\begin{align*}
		&\mathcal R^{S\text{-}CBC}(x) = \\
		&\begin{bmatrix} -\kappa P^{-1} & 0  \notag\\ * & (1+\rho) \underbrace{\left[\begin{array}{c} \mathcal{J}(x) \notag\\
					\mathcal{G}(x)\mathcal K(x)\end{array}\right]}_{\Lambda(x)} P^{-1}\underbrace{\left[\begin{array}{c} \mathcal{J}(x) \notag\\
					\mathcal{G}(x)\mathcal K(x)\end{array}\right]^\top}_{\Lambda(x)^\top} \end{bmatrix}\!\!.
	\end{align*}
	{\bf Construction of the data-conformity constraints.} While $\Phi$ in~\eqref{Eq_CBC3_new} is unknown, Lemma~\ref{Lem1} offers $T$ quadratic matrix inequalities that characterize its behavior through data-conformity constraints. In particular, the collected data in \eqref{Eq_ST}, for $j= 1,\dots,T,$ implies that
	\begin{align}\notag
		\Xright_j^i  =& A\mathbb{F}_j^i  +\!B{\mathbb{G}}_j^i  +\mathbb{Z}_j^i   = \Phi  \mathbb{H}_j^i  +\mathbb{Z}_j^i , \\\label{New43}
		&\,\, \text{where }\,\, \Phi = [A~~~B], \quad\mathbb{H}_j^i   \!=\! \left[\begin{array}{c} \mathbb{F}_j^i   \\
			{\mathbb{G}}_j^i   \end{array}\right]\!\!. 
	\end{align}
	Subsequently, $\mathbb{Z}_j^i    \!=\! \Xright_j^i   \!-\! \Phi  \mathbb{H}_j^i  $. 
	According to equation (\ref{Eq_DC}), for \emph{each fixed} \( j \in \{1, \dots, T\} \), we define the event
	\[
	\mathcal{E}_j := \Big\{
	\frac{1}{N} \sum_{i=1}^N \mathbb{Z}^i_{j} \mathbb{Z}^{i\top}_{j} 
	\preceq \Gamma_\Sigma + \Gamma_\mu + \epsilon \mathds{I}_n
	\Big\}.
	\]
	Under the results of Lemma \ref{Lem1}, \( \mathbb{P}(\mathcal{E}_j) \geq 1 - \bar\beta_2 \) for each \( j \).  
	We are interested in the joint event \( \bigcap_{j=1}^T \mathcal{E}_j \).  
	Using the Boole's inequality \cite{boole1847mathematical}, we obtain
	\[
	\mathbb{P} \Big( \bigcup_{j=1}^T \mathcal{\bar E}_j \Big) 
	\leq \sum_{j=1}^T \mathbb{P}(\mathcal{\bar E}_j) 
	\leq T \bar\beta_2,
	\]
	with $\mathcal{\bar E}_j$ being the complement of $\mathcal{E}_j$. Therefore, the probability that all \( T \) events occur simultaneously is at least
	\begin{align}\label{ooo}
		\mathbb{P} \Big( \bigcap_{j=1}^T \mathcal{E}_j \Big) 
		= 1 - \mathbb{P} \Big( \bigcup_{j=1}^T \mathcal{\bar E}_j \Big) 
		\geq 1 - \overbrace{T \bar\beta_2}^{\beta_2}.
	\end{align}
	According to the closed-form data-based representation of $ \mathbb{Z}^i_j $ in \eqref{New43} as  $\mathbb{Z}_j^i = \Xright_j^i  -  \Phi  \mathbb{H}_j^i $, and probabilistic inequality (\ref{ooo}), one has
	\begin{align}
		\Gamma_\Sigma&+ \Gamma_\mu +\epsilon\mathds I_{n} \succeq  \frac{1}{N} \sum_{i=1}^N \mathbb{Z}^i_j \mathbb{Z}^{i\top}_j \nonumber \\ 
		=&\, \frac{1}{N} \sum_{i=1}^N(\Xright^i_j  \!-\! \Phi  \mathbb{H}^i_j )(\Xright^i_j \! -\! \Phi  \mathbb{H}^i_j )^\top   \nonumber\\
		=&\, \frac{1}{N} \sum_{i=1}^N\big(\Phi  \mathbb{H}^i_j  \mathbb{H}_j ^{i\top} \Phi^\top \!\!-\! \Phi  \mathbb{H}^i_j \Xright_j^{i\top} \!\!-\!\Xright^i_j  \mathbb{H}_j^{i\top}  \Phi^\top \!\!+\! \Xright^i_j  \Xright_j^{i\top}\big), \label{Eq_DC111}
	\end{align}
	which yields $T$ matrix inequalities for $\Phi$, holding true with the confidence of at least $1 - \beta_2$. Accordingly, \eqref{Eq_DC111} can be recast as
	\begin{gather} \label{Eq_DC_new}
		\mathcal{Q}_{\Phi}^{DC_j} :=\thetaTildeI^\top \mathcal R^{DC_j} \thetaTildeI \preceq 0,
	\end{gather}
	with $\mathcal R^{DC_j}$ as in \eqref{Eq_PI_new1}.
	
	\noindent {\bf S-lemma-based synthesis condition.}
	By applying S-lemma \cite{9308978}, to enforce \eqref{Eq_CBC3_new} where \eqref{Eq_DC_new} is fulfilled, it is sufficient to show that there exists $\alpha_{j=1,\dots,T}(x):\mathbb{R}^n \to  \mathbb{R}_0^+$ such that
	\begin{align} \label{Eq_Slemma} 
		& \mathcal R^{S\text{-}CBC}(x) - \sum_{j=1}^{T}\alpha_j(x)\mathcal R^{DC_j} \preceq 0.
	\end{align}
	As there is a bilinear coupling between $P^{-1}$ and $\mathcal{K}(x)$ in~\eqref{Eq_CBC3_new} within the term $\Lambda(x) P^{-1} \Lambda(x)^\top$, this interaction can be reformulated using a dilation technique based on the Schur complement \cite{caverly2019lmi}, allowing us to express inequality~\eqref{Eq_Slemma} equivalently as \eqref{Eq_CBC3_theorem}, \emph{i.e.,}
	\begin{align*}
		\eqref{Eq_Slemma}\Leftrightarrow   \eqref{Eq_CBC3_theorem}.
	\end{align*}
	Since \eqref{Eq_CBC3_theorem} holds, condition \eqref{Eq_Slemma} is automatically satisfied. By selecting appropriate scalar values for $\rho > 0$ and $\kappa \in (0,1)$, the expression in \eqref{Eq_CBC3_theorem} becomes a linear matrix inequality (LMI) in the design variables $\bar P$, $\bar{\mathcal{K}}(x)$, and $\alpha_{j=1,\dots,T}(x)$. 
	
	\noindent  {\bf Satisfaction of the initial and unsafe set conditions.} We proceed with showing that satisfying conditions \eqref{Eq_CBC1_theorem} and \eqref{Eq_CBC2_theorem} implies the fulfillment of conditions \eqref{Eq_CBC1} and \eqref{Eq_CBC2}, respectively.
	Given that $P\succ 0$ and $\eta, \delta \in \mathbb{R}^+,$ by leveraging Schur complement \cite{caverly2019lmi}, one can verify that
	\begin{align*}
		\eta - x^\top Px \geq  0~&\Leftrightarrow ~ P^{-1} - \eta^{-1}xx^\top \succeq 0,
	\end{align*}
	which is equivalent to \eqref{Eq_CBC1}. Similarly, according to Schur complement, one has 
	\begin{align*}
		\delta - x^\top Px >  0~&\Leftrightarrow ~ P^{-1} - \delta^{-1}xx^\top \succ0,
	\end{align*}
	demonstrating that their complements are also equivalent:
	\begin{align}\label{New76}
		\delta - x^\top Px \ngtr  0~&\Leftrightarrow ~ P^{-1} - \delta^{-1}xx^\top \nsucceq 0.
	\end{align}
	Observe that the left-hand side of \eqref{New76} matches the expression in \eqref{Eq_CBC2}. As a result, it is necessary to ensure that $P^{-1} - \delta^{-1} xx^\top \nsucceq 0$. Since the matrix inequality $\nsucceq 0$ indicates that the matrix is either negative semi-definite or indefinite, with the latter being difficult to certify, a conservative relaxation using `$\preceq$' is adopted. Together with the definitions of $X_\eta$ and $X_\delta$ in \eqref{Eq_CBC1_theorem} and \eqref{Eq_CBC2_theorem}, this ensures that satisfying \eqref{Eq_CBC1_theorem} and \eqref{Eq_CBC2_theorem} is sufficient to meet the original conditions \eqref{Eq_CBC1} and \eqref{Eq_CBC2}.
	
	\noindent {\bf Probabilistic guarantee.} As the final step of the proof, we demonstrate that the theorem's conclusion holds with a confidence level of $1 - \beta_2$. By defining events
	\begin{align*}
		\mathcal{E}_1 = \big\{\text{inequality}~ \eqref{Eq_CBC3_new}\big\},\quad \mathcal{E}_2 = \big\{ \text{inequality}~\eqref{Eq_DC_new}\big\},
	\end{align*}
	one has $\PP\{\mathcal{\bar E}_1\}=0$ since $\mathcal{E}_1$ is a deterministic inequality and holds true
	and  $\PP\{\mathcal{\bar E}_2\}\leq \beta_2,$ where $\mathcal{\bar E}_1$ and $\mathcal{\bar E}_2$ are the complement of $\mathcal{E}_1$ and $\mathcal{E}_2$, respectively. The goal is to characterize the joint probability that both  events $\mathcal E_1$ and $\mathcal{E}_2$ hold simultaneously:
	\begin{align*}
		\PP(\mathcal E_1\cap \mathcal{E}_2)=1-\PP\big(\mathcal{\bar E}_1\cup \mathcal{\bar E}_2\big).
	\end{align*}
	Since $\PP\big(\mathcal{\bar E}_1\cup \mathcal{\bar E}_2\big)\leq\PP(\mathcal{\bar E}_1)+\PP(\mathcal{\bar E}_2)$,
	one has 
	\begin{align}\label{Eq:18}
		&\PP(\mathcal E_1\cap \mathcal{E}_2)\geq 1-\PP\big(\mathcal{\bar E}_1\big)-\PP\big(\mathcal{\bar E}_2\big)
		\geq 1-\beta_2.
	\end{align}
	This establishes that inequalities \eqref{Eq_CBC3_new} and \eqref{Eq_DC_new}, and consequently condition~\eqref{Eq_Slemma}, hold with a confidence level of at least $1 - \beta_2$, thereby concluding the proof. \hfill $\blacksquare$

\section{Appendix: Simulation results}\label{app: simulations}

\subsection{Lorenz System}

{\bf Dynamics.}
\begin{equation*}
	\Upsilon\!: \begin{cases}
		{{x}}_1^+=x_1 + \tau(10x_2 -10x_1) + \varsigma_1 \\
		{{x}}_2^+=x_2 + \tau(28x_1 - x_2- x_1 x_2+u)  + \varsigma_2\\
		{{x}}_3^+=x_3 + \tau(x_1x_3 - \frac{8}{3}x_3)  + \varsigma_3
	\end{cases}
\end{equation*}
{\bf Extended dictionary.} $\mathcal{F}(x) = [ x_1; x_2; x_3; x_1x_2; x_1x_3; x_2x_3; x_1^2; x_2^2; x_3^2 ]$\vspace{0.2cm}

{\bf  Prescribed regions of interest.}
$X = [-10, 10]^3$, $X_\eta = [0, 1.5] \times [-1.5, 1.5]^2$, $X_\delta = [-10, -6]^3 \cup [6, 10]^3$\vspace{0.2cm}

{\bf Designed S-CBC matrix $P$.}
\begin{equation*}
	P = \begin{bmatrix}
		33450.3 & 3339.99 & 656.04\\ *& 38333 & 1417.8\\ * & * & 34290.3 
	\end{bmatrix}
\end{equation*}

{\bf Designed controller.}
\begin{align}\notag
	u &= 
	-1.4389\,x_{1}^2 - 1.4844\,x_{1}x_{2} - 0.50617\,x_{1}x_{3}\\\label{sdhg}
	&\quad + 1.7463\,x_{2}^2 + 0.62559\,x_{2}x_{3} + 0.1802\,x_{3}^2\\\notag
	&\quad - 1.9406\,x_{1} - 28.1245\,x_{2} + 0.1565\,x_{3}
\end{align}

{\bf Designed level sets.}
$\eta = 2.72 \times 10^{5}$, $\delta =  4.02 \times 10^{6}$\vspace{0.2cm}

{\bf Probabilistic safety guarantee.}
\begin{align*}
	\PP\Big\{\PP\big\{\Upsilon\vDash \mathbb{S}\big\}\geq 0.92\Big\}\geq 0.995
\end{align*}

\subsection{Chen System}

{\bf Dynamics.}
\begin{align*}
	\Upsilon\!: \begin{cases}
		x^+_{1}= x_1 + \tau\,(35\, x_{2}-35\, x_{1}) +\varsigma_1\\
		x^+_{2} = x_2 + \tau\,(-7\, x_{1}+28\,x_{2}-x_{1} x_{3} +u) +\varsigma_2\\
		x^+_{3} = x_3 + \tau\,(x_{1} x_{2}-3\, x_{3}) +\varsigma_3
	\end{cases}
\end{align*}

{\bf Extended dictionary.} $\mathcal{F}(x) = [ x_1; x_2; x_3; x_1x_2; x_1x_3; x_2x_3]$\vspace{0.2cm}

{\bf  Prescribed regions of interest.} $X = [-10, 10]^3$, $X_\eta = [-1, 1]^3$, $X_\delta = [-10, -6]^3 \cup [6, 10]^3$\vspace{0.2cm}

{\bf Designed S-CBC matrix $P$.} \begin{equation*}
	P = \begin{bmatrix}
		65587.8 & 20628.6 & 486.86\\ * & 103532 & 984.94\\ * & * & 60375.3 
	\end{bmatrix}
\end{equation*}

{\bf Designed controller.} 
\begin{align}\notag
	u &= 
	0.016422\,x_{1}^2 - 0.12604\,x_{1}x_{2} - 0.48379\,x_{1}x_{3}\\\label{hdjfh}
	&\quad + 0.0019506\,x_{2}^2 + 2.2790\,x_{2}x_{3} + 0.030116\,x_{3}^2\\\notag
	&\quad + 4.9195\,x_{1} - 32.3368\,x_{2} + 0.38473\,x_{3}
\end{align}

{\bf Designed level sets.} $\eta = 2.67 \times 10^{5}$, $\delta = 8.31 \times 10^{6}$\vspace{0.2cm}

{\bf Probabilistic safety guarantee.}
\begin{align*}
	\PP\Big\{\PP\big\{\Upsilon\vDash \mathbb{S}\big\}\geq 0.95 \Big\}\geq 1 - 35\!\times\! 10^{-5}
\end{align*}

\subsection{Spacecraft System}

{\bf Dynamics.}
\begin{equation*}
	\Upsilon\!: \begin{cases}
		x^+_{1} = x_1 + \tau\, (\frac{J_{2} - J_{3}}{J_{1}}\, x_{2}\, x_{3} + \frac{1}{J_{1}}\, u_{1}) +\varsigma_1 \\
		x^+_{2} = x_2 + \tau \,(\frac{J_{3} - J_{1}}{J_{2}}\,x_{1}\, x_{3} + \frac{1}{J_{2}}\, u_{2}) + \varsigma_2\\
		x^+_{3} = x_3 + \tau \, (\frac{J_{1} - J_{2}}{J_{3}}\, x_{1}\, x_{2} + \frac{1}{J_{3}}\, u_{3}) + \varsigma_3
	\end{cases}
\end{equation*}
{\bf Extended Dictionary.}  $\mathcal{F}(x) =[x_{1};x_{2};x_{3};x_{1}x_{2};x_{1}x_{3};x_{2}x_{3}; x^2_{1}]$\vspace{0.2cm}

{\bf  Prescribed regions of interest.}
$X = [-10, 10]^3$, $X_\eta = [-1, 1]^3$, $X_\delta = [-10, -6]^3 \cup [6, 10]^3$\vspace{0.2cm}

{\bf Designed S-CBC matrix $P$.}
\begin{equation*}
	P = \begin{bmatrix}
		2021150 & -18.70 & 548.94\\ * & 1102220 & 162.13\\ * & * & 5534420
	\end{bmatrix}
\end{equation*}

{\bf Designed controller.}
\begin{align*}
	u_1 &= 
	-0.54439\,x_{1}^2 + 0.0041722\,x_{1}x_{2} + 0.026666\,x_{1}x_{3}\\
	&\quad + 9.1059 \!\times\! 10^{-6}\,x_{2}^2 \!-\! 0.0043456\,x_{2}x_{3} \!+\! 5.7216 \!\times\! 10^{-6}\,x_{3}^2\\
	&\quad + 0.14816\,x_{1} + 0.11343\,x_{2} - 0.0014198\,x_{3}\\[1em]
	u_2 &= 
	-0.31636\,x_{1}^2 + 0.011794\,x_{1}x_{2} + 0.017499\,x_{1}x_{3}\\
	&\quad + 6.4658 \!\times\! 10^{-7}\,x_{2}^2 \!-\! 0.0026077\,x_{2}x_{3} \!+\! 3.6263 \!\times\! 10^{-6}\,x_{3}^2\\
	&\quad + 0.025553\,x_{1} + 0.037923\,x_{2} + 0.0001751\,x_{3}\\[1em]
	u_3 &= 
	0.31662\,x_{1}^2 - 0.012773\,x_{1}x_{2} - 0.016738\,x_{1}x_{3}\\
	&\quad + 1.7739 \!\times\! 10^{-6}\,x_{2}^2 \!+\! 0.0024458\,x_{2}x_{3} \!-\! 3.9601 \!\times\! 10^{-6}\,x_{3}^2\\
	&\quad - 0.017964\,x_{1} - 0.053823\,x_{2} - 0.00033808\,x_{3}
\end{align*}

{\bf Designed level sets.}
$\eta =  1.66 \times 10^{7}$, $\delta = 2.98 \times 10^{8}$\vspace{0.2cm}

{\bf Probabilistic safety guarantee.}
\begin{align*}
	\PP\Big\{\PP\big\{\Upsilon\vDash \mathbb{S}\big\}\geq 0.93\Big\}\geq 0.96
\end{align*}

\begin{IEEEbiography}[{\includegraphics[width=1in,height=1.25in,clip,keepaspectratio]{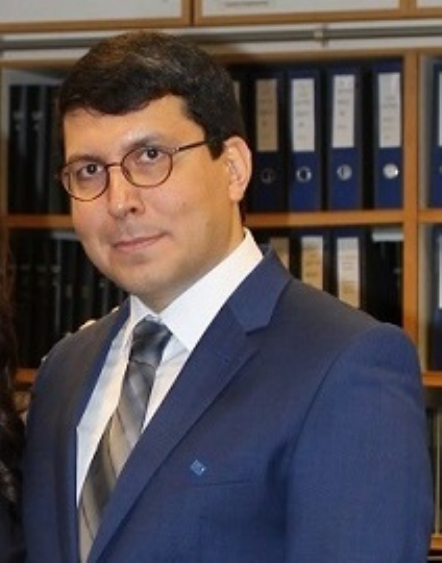}}]{Abolfazl Lavaei}~(M'17--SM'22) is an Assistant Professor in the School of Computing at Newcastle University, United Kingdom. Between January 2021 and July 2022, he was a Postdoctoral Associate in the Institute for Dynamic Systems and Control at ETH Zurich, Switzerland. He was also a Postdoctoral Researcher in the Department of Computer Science at LMU Munich, Germany, between November 2019 and January 2021. He received the Ph.D. degree in Electrical Engineering from the Technical University of Munich (TUM), Germany, in 2019. He obtained the M.Sc. degree in Aerospace Engineering with specialization in Flight Dynamics and Control from the University of Tehran (UT), Iran, in 2014. He is the recipient of several international awards in the acknowledgment of his work including  Best Repeatability Prize (Finalist) at the ACM HSCC 2025, IFAC ADHS 2024, and IFAC ADHS 2021, HSCC Best Demo/Poster Awards 2022 and 2020, IFAC Young Author Award Finalist 2019, and Best Graduate Student Award 2014 at University of Tehran with the full GPA (20/20). His line of research primarily focuses on the intersection of Control Theory, Formal Methods, and Statistical Learning Theory.
\end{IEEEbiography}

\end{document}